\documentclass[a4paper]{article}

\usepackage[utf8]{inputenc}
\usepackage[T1]{fontenc}
\usepackage[USenglish]{babel}
\usepackage{microtype} 

\usepackage{amsfonts}
\usepackage{amsmath}
\usepackage{amssymb}
\usepackage{amsthm}
\usepackage{mathtools}
\usepackage{bm} 

\usepackage{hyperref}
\usepackage{cleveref} 
\usepackage{csquotes}
\usepackage[shortlabels]{enumitem}
\setenumerate[1]{label=(\arabic*)}
\usepackage{multicol}
\usepackage{tikz}
\usetikzlibrary{arrows.meta,
                calc,
                decorations,
                decorations.pathreplacing,
                positioning,
                shapes}
\usepackage{todonotes}

\usepackage{url}

\tikzset{%
  vertex/.style={circle, fill=black, inner sep=2pt},
  hypervertex/.style={circle, fill=black, inner sep=2pt},
  hyperedge/.style={rectangle, fill=black, inner sep=2pt},
  hyperedgeedge/.style={thick},
  smallvertex/.style={circle, fill=black, inner sep=1.3pt},
  edge/.style={thick, -{Latex[length=1.5mm, width=1.5mm]}},
  shortcut/.style={dashed, thick, -{Latex[length=1.5mm, width=1.5mm]}},
  undirectedEdge/.style={thick},
  layer/.style={ellipse,
                draw,
                minimum width=4cm,
                minimum height=1.5cm},
  decoedge/.style={-{Latex[length=1mm, width=1mm]}},
  decoshortcut/.style={dashed, -{Latex[length=1mm, width=1mm]}},
  dist/.style={2cm},
  mapping/.style={shorten >=2pt, shorten <=2pt, -latex, thick, bend left=15},
  iso/.style={shorten >=2pt, shorten <=2pt, -latex, very thick,},
  instruction/.style={-latex, ultra thick, dashed, gray},
  mappingHP/.style={densely dotted, semithick, black!60!white, shorten >=2pt, shorten <=2pt, ->,bend left=15}
}

\newcommand{\Powerset}[1]{2^{#1}}

\renewcommand{\emptyset}{\varnothing}

\newcommand{\restr}[2]{#1|_{#2}}

\newcommand{\Neighbors}[2][]{N_{#1}(#2)}

\newcommand{\Direction}[2]{``#1~$\! \Rightarrow \!$~#2''}

\renewcommand{\phi}{\varphi}

\newcommand{\N}{\mathbb{N}}

\newcommand{\R}{\mathbb{R}}

\newcommand{\cupdot}{\;\dot{\cup}\;}

\DeclarePairedDelimiter{\abs}{\lvert}{\rvert}

\newcommand{\DAGs}{\mathcal{A}}

\newcommand{\Trees}{\mathcal{T}}
\newcommand{\ColoredTrees}{\mathcal{C}\hspace{-0.5pt}\mathcal{T}}
\newcommand{\Graphs}{\mathcal{G}}

\newcommand{\SomeClass}{\mathcal{F}}

\newcommand{\iso}{\cong}

\newcommand{\Hom}{\mathsf{Hom}}

\newcommand{\HOMVector}{\mathsf{HOM}}

\newcommand{\HOM}[2]{\mathsf{HOM}_{#1}(#2)}

\newcommand{\Fam}[2]{\left(#1\right)_{#2}}
\newcommand{\Set}[1]{\{\, #1 \,\}}
\newcommand{\MultiSet}[1]{\{\!\!\{\, #1 \,\}\!\!\}}

\DeclareFontFamily{U}{mathx}{\hyphenchar\font45}
\DeclareFontShape{U}{mathx}{m}{n}{ <-> mathx10 }{}
\DeclareSymbolFont{mathx}{U}{mathx}{m}{n}
\DeclareFontSubstitution{U}{mathx}{m}{n}
\DeclareMathAccent{\widebar}{\mathalpha}{mathx}{"73}

\makeatletter
\newcommand{\cwidebar}[2][0]{{\mathpalette\@cwidebar{{#1}{#2}}}}
\newcommand{\@cwidebar}[2]{\@cwideb@r{#1}#2}
\newcommand{\@cwideb@r}[3]{%
  \sbox\z@{$\m@th#1\mkern-#2mu#3\mkern#2mu$}%
  \widebar{\box\z@}%
}
\makeatother

\newcommand{\TensorProduct}{\otimes}

\newcommand{\TournamentOf}[1]{\overrightarrow{K_{#1}}}

\newcommand{\NumbersTo}[1]{\left[#1\right]}

\newcommand{\LeafAddingI}{\mathsf{LeafAddInHom}}

\newcommand{\LeavesOf}[2][]{L_{#1}(#2)}
\newcommand{\KnOf}[1]{K_{#1}}

\newcommand{\QuotientGraph}[2]{#1/\!{#2}}

\newcommand{\ImageOf}[1]{im(#1)}

\def\EdgeDistance{1.5}%
\def\VertexDistance{1.3}%
\def\NotSoLongBend{40}%
\def\TadLongerShorterBend{28}%

\newcommand{\MergeEquivSpaceOf}[1]{\equiv_{#1}}
\newcommand{\MergeEquivOf}[1]{\equiv_{#1}\!}

\newcommand{\MergeEquivClassesOf}[2]{#1/\!{\MergeEquivOf{#2}}}
\newcommand{\EqClassOf}[1]{[#1]}

\newcommand{\EdgeFunOf}[1]{f_{#1}}

\newcommand{\AutGroup}{\text{Aut}}

\newcommand{\CTs}{\mathcal{C\!T}}
\newcommand{\BACHs}{\mathcal{B\!A}}
\newcommand{\SBACHs}{\mathcal{S\!B\!A}}

\newcommand{\IHom}{\mathsf{InHom}}
\newcommand{\IHOM}[2]{\mathsf{InHOM}_{#1}(#2)}

\newcommand{\LOCINJIHOM}[2]{\mathsf{LoInjInHOM}_{#1}(#2)}
\newcommand{\LOCINJHOM}[2]{\mathsf{LoInjHOM}_{#1}(#2)}

\newcommand{\LocInjIHom}{\mathsf{LoInjInHom}}
\newcommand{\LocInjHom}{\mathsf{LoInjHom}}
\newcommand{\EMergHom}{\mathsf{LoMeHom}}

\newcommand{\Aut}{\mathsf{Aut}}

\newcommand{\IncGraph}[1]{I({#1})}
\newcommand{\CIncGraph}[1]{I_{\hspace{-0.05em}c}({#1})}
\newcommand{\ColRef}[2][]{C^{#1}_{#2}}
\newcommand{\HPColRef}[2][]{H\!C^{#1}_{#2}}

\newcommand{\DegreeOf}[2][]{d_{#1}(#2)}
\newcommand{\DegreeSeqOf}[1]{\bar{d}(#1)}
\newcommand{\IDegreeOf}[2]{d_{#1}(#2)}

\newcommand{\IHomDegree}[3]{\IHom(#1,#2,#3)}

\newcommand{\IHomDegreeSeq}[3]{\IHom(#1,#2,#3)}
\newcommand{\DegreeSeqsOf}[1]{D_{#1}}

\providecommand{\keywords}[1]
{
   {\small	
  \textbf{Keywords ---} #1}
}

\theoremstyle{definition}
\newtheorem{definition}{Definition}

\theoremstyle{plain}

\theoremstyle{plain}

\theoremstyle{plain}
\newtheorem{lemma}[definition]{Lemma}

\theoremstyle{plain}
\newtheorem{theorem}[definition]{Theorem}

\theoremstyle{plain}
\newtheorem{corollary}[definition]{Corollary}

\begin{document}

\title{Color Refinement, Homomorphisms, and Hypergraphs}
\author{Jan Böker\\%
        RWTH Aachen University, Aachen, Germany\\%
        \href{mailto:boeker@informatik.rwth-aachen.de}{boeker@informatik.rwth-aachen.de}\\%
        ORCID: 0000-0003-4584-121X}
\date{}

\maketitle

\begin{abstract}
    Recent results show that the structural similarity of graphs can be characterized by counting homomorphisms to them:
    the Tree Theorem states that the well-known color-refinement algorithm does not distinguish two graphs $G$ and $H$ if and only if, for every tree $T$, the number of homomorphisms $\Hom(T, G)$ from $T$ to $G$ is equal to the corresponding number $\Hom(T, H)$ from $T$ to $H$ (Dell, Grohe, Rattan $2018$).
    We show how this approach transfers to hypergraphs by introducing a generalization of color refinement.
    We prove that it does not distinguish two hypergraphs $G$ and $H$ if and only if, for every connected Berge-acyclic hypergraph $B$, we have $\Hom(B, G) = \Hom(B, H)$.
    To this end, we show how homomorphisms of hypergraphs and of a colored variant of their incidence graphs are related to each other.
    This reduces the above statement to one about vertex-colored graphs.
\end{abstract}


\section{Introduction}
\label{sec:in}

A result by Lovász \cite{Lovasz1967} states that a graph can be characterized up to isomorphism by counting homomorphisms from all graphs to it, i.e., two graphs $G$ and $H$ are isomorphic if and only if, for every graph $F$, the number of homomorphisms $\Hom(F,G)$ from $F$ to $G$ is equal to the number of homomorphisms $\Hom(F,H)$ from $F$ to $H$.
Equivalently, using the notion of the \textit{homomorphism vector} $\HOMVector(G) \coloneqq \Fam{\Hom(F,G)}{F \in \Graphs}$ of $G$, where $\Graphs$ denotes the class of all graphs, we have that two graphs $G$ and $H$ are isomorphic if and only if their homomorphism vectors $\HOMVector(G)$ and $\HOMVector(H)$ are equal.
However, the problem of computing the entries of a homomorphism vector is \#P-complete as it generalizes some well-known counting problems \cite[Section~$5.1$]{Lovasz2012}.
Hence, Dell, Grohe, and Rattan \cite{Dell2018} considered restrictions $\HOM{\SomeClass}{G} \coloneqq \Fam{\Hom(F,G)}{F \in \SomeClass}$ of homomorphism vectors to classes of graphs~$\SomeClass$ for which these entries can be computed efficiently.
This yields some surprisingly clean results, e.g., for the class $\Trees$ of all trees, the \textit{Tree Theorem} states that the homomorphism vectors $\HOM{\Trees}{G}$ and $\HOM{\Trees}{H}$ of two graphs $G$ and $H$ are equal if and only if $G$ and $H$ are not distinguished by \textit{color refinement}, a well-known heuristic algorithm for distinguishing non-isomorphic graphs (e.g.\ \cite{Grohe2017}).

\enquote{Graph matching} is a term used in machine learning for the problem of measuring the similarity of graphs (e.g., \cite{Conte2004}), where it has its applications in pattern recognition.
However, there is no universally agreed-upon notion of similarity, and a popular notion, the graph edit distance, describing the cost of transforming one graph into another by adding and deleting vertices and edges, is not only hard to compute but also does not reflect the structural similarity of two graphs very well \cite[Section~$1.5.1$]{Lovasz2012}.
Restricted homomorphism vectors offer an alternative way of comparing the structural similarity of graphs since, after suitably scaling them, they can be compared using standard vector norms.
As demonstrated in \cite{Dell2018}, one can also define an inner product on these homomorphism vectors, which yields a mapping that is known as a \textit{graph kernel} in machine learning (e.g., \cite{Vishwanathan2010}).
Graph kernels can be used to perform classification on graphs, and to this end, should capture the similarity of graphs well while still being efficiently computable.
Similarly to homomorphism vectors, state-of-the-art graph kernels are usually based on counting certain patterns in graphs, e.g., walks or subtrees.

The original observation by Lovász \cite{Lovasz1967}, stating that a graph can be characterized up to isomorphism by counting homomorphisms from all graphs, dates back to the $1960$s and has led to the theory of graph limits in the recent past \cite{Lovasz2012}.
Only very recently, the importance of homomorphism counts for many graph-related counting problems has been recognized \cite{Curticapean2017}:
for example, subgraph counts are just linear combinations of homomorphism counts.
Even more recent is the approach of characterizing the structural similarity of graphs by counting homomorphisms from restricted classes of graphs \cite{Dell2018}, which shows that well-known characterizations, e.g., the color-refinement algorithm, can also be stated in terms of homomorphism counts.

\subsection{Overview}
\label{sec:overview}

Color refinement is a simple and efficient but incomplete algorithm for distinguishing non-isomorphic graphs.
The algorithm iteratively computes a coloring of the vertices of a graph, and we say that color refinement \textit{distinguishes} two graphs if it computes different color patterns for them.
The Tree Theorem~\cite{Dell2018} states that color refinement can be characterized by counting homomorphisms from trees, i.e., for all graphs $G$ and $H$, we have $\HOM{\Trees}{G} = \HOM{\Trees}{H}$ if and only if color refinement does not distinguish $G$ and $H$.
By making use of the initial coloring, color refinement can easily be adapted to vertex-colored graphs.
This enables a straight-forward generalization of the Tree Theorem by counting (color-respecting) homomorphisms from vertex-colored trees to vertex-colored graphs.
Formally, if we let $\ColoredTrees$ denote the class of all vertex-colored trees, then for all vertex-colored graphs $G$ and $H$, we have $\HOM{\ColoredTrees}{G} = \HOM{\ColoredTrees}{H}$ if and only if color refinement does not distinguish $G$ and $H$.
We refer to this generalization as the \textit{Colored Tree Theorem}.

A possible (although rather conservative) generalization of the notion of a tree to hypergraphs is that of a \textit{connected} \textit{Berge-acyclic} hypergraph.
A hypergraph is called connected and Berge-acyclic if its incidence graph is connected and acyclic, respectively.
Similarly to the case of (vertex-colored) graphs, we obtain a surprisingly clean answer when counting homomorphisms from hypergraphs in the class $\BACHs$ of connected Berge-acyclic hypergraphs.
\begin{theorem}
    \label{th:hy:countingBACHHoms}
    For all hypergraphs $G$ and $H$, the following are equivalent:
    \begin{enumerate}
        \item $\HOM{\BACHs}{G} = \HOM{\BACHs}{H}$.
        \item Color refinement does not distinguish $G$ and $H$.
    \end{enumerate}
\end{theorem}

Of course, color refinement in the usual sense is only defined on (vertex-colored) graphs, which is why we propose a generalization of it to hypergraphs in \Cref{sec:hypergraphCR};
\Cref{th:hy:countingBACHHoms} refers to this generalization.
As this generalization turns out to be equivalent to the usual color-refinement algorithm applied to a colored variant of a hypergraph's incidence graph, we are able to \enquote{reduce} \Cref{th:hy:countingBACHHoms} to the Colored Tree Theorem instead of adapting the proof of \cite{Dell2018}.
Here, the interesting (and laborious) part is to show how homomorphisms between hypergraphs are related to homomorphisms between their colored incidence graphs and how counts of these can be obtained from each other.
This leads to the notion of an \textit{incidence homomorphism} between hypergraphs in \Cref{sec:homomorphismsFromBACHs}, which is used to prove \Cref{th:hy:countingBACHHoms} in \Cref{sec:hy:BACHs}.
Our approach does not only directly generalize to hypergraphs that possibly have parallel edges, but is also simplified by doing so;
we nevertheless obtain the corresponding statement about simple hypergraphs as a corollary in \Cref{sec:simplehypergraphs}.

With \Cref{th:hy:countingBACHHoms}, one might wonder how the Tree Theorem generalizes to directed graphs.
An obvious candidate for a class of directed graphs to count homomorphisms from is the class of connected directed acyclic graphs (DAGs) since a connected DAG can be seen as the directed concept corresponding to a tree.
Surprisingly, counting homomorphisms from DAGs is already too expressive and characterizes an arbitrary directed graph up to isomorphism.
This result is already implicit in the second homomorphism-related work of Lovász \cite{Lovasz1971}, which is concerned with the \textit{cancellation law} among finite relational structures, and we briefly revisit it in \Cref{sec:di}.

\subsection{Preliminaries}

$\N$ denotes the set of non-negative integers.
For $n \in \N$, we let $\NumbersTo{n} \coloneqq \Set{1, \dots, n}$.
A multiset is denoted using the notation $\MultiSet{0, 1, 1}$.
All relational structures that we consider are finite, and we use standard graph-theoretic terminology and notation without explicitly introducing it, e.g., for any graph-like structure $G$, the sets of its vertices and edges are denoted by $V(G)$ and $E(G)$, respectively.
Unless explicitly specified otherwise, the terms \textit{graph} and \textit{directed graph} refer to simple graphs and simple directed graphs, respectively, while for the sake of brevity, the term \textit{hypergraph} is used for hypergraphs that may have parallel edges.
Formally, a hypergraph is a tuple $G = (V, E, f)$ where $V$ is a set of vertices, $E$ a set of edges, and $f \colon E \to \Powerset{V} \setminus \Set{\emptyset}$ the incidence function assigning a non-empty set of vertices to every edge, where we usually write $\EdgeFunOf{G}$ to denote $f$.
If $f$ is injective, i.e., if $G$ does not have parallel edges, then we call $G$ a \textit{simple hypergraph}.
The incidence graph of a hypergraph $G$ is the bipartite graph $\IncGraph{G}$ with $V(\IncGraph{G}) \coloneqq V(G) \cupdot E(G)$ and $E(\IncGraph{G}) \coloneqq \Set{ve \mid v \in \EdgeFunOf{G}(e) \text{ for } e \in E(G)}$.

We work with \textit{infinite matrices}, which are functions $A \colon I \times J \to \R$ where $I$ and $J$ are countable and locally finite posets.
The product $A \cdot B \colon I \times J \to \R$ of two infinite matrices $A \colon I \times K \to \R$ and $B \colon K \times J \to \R$ is defined via $(A \cdot B)_{ij} \coloneqq \sum_{k \in K} A_{ik} \cdot B_{kj}$ for all $i \in I,\, j \in J$ as long as these sums are finite; otherwise, we leave it undefined, which means that this product is not associative, and we follow the convention that this operator is right-associative to reduce the amount of needed parentheses.
An infinite matrix $A$ is called lower triangular and upper triangular if we have $A_{ij} = 0$ for all $i,j$ with $j \not\le i$ and $A_{ij} = 0$ for all $i,j$ with $i \not\le j$, respectively.
As in the finite case, forward substitution yields that lower and upper triangular infinite matrices with non-zero diagonal entries have left inverses \cite{Dell2018} that again are lower and upper triangular, respectively.
For simplicity, we usually refer to infinite matrices just as matrices.

A homomorphism from a hypergraph $F$ to a hypergraph $G$ is a pair $(h_V, h_E)$ of mappings $h_V \colon V(F) \to V(G)$ and $h_E \colon E(F) \to E(G)$ such that we have $h_V(\EdgeFunOf{F}(e)) = \EdgeFunOf{G}(h_E(e))$ for every $e \in E(F)$, and $\Hom(F,G)$ denotes the number of homomorphisms from $F$ to $G$.
For a hypergraph~$G$, its homomorphism vector is denoted by $\HOMVector(G)$, and the restriction of $\HOMVector(G)$ to a class of hypergraphs~$\SomeClass$ is denoted by $\HOM{\SomeClass}{G}$.
For every isomorphism class of hypergraphs, we fix a representative and call it the \textit{isomorphism type} of the hypergraphs in the class.
We view $\Hom$ as an infinite matrix indexed by the isomorphism types, which are sorted by the sums of their numbers of vertices and edges, where ties are resolved arbitrarily.
Then, for a hypergraph $G$, its homomorphism vector $\HOMVector(G)$ can be viewed as a column of $\Hom$.
We use similar notation for other types of mappings without explicitly introducing it.

Since, to count homomorphisms from a non-connected graph, one can count homomorphisms from its components instead, we usually restrict ourselves to homomorphism counts from connected graphs.
The same holds for directed graphs and hypergraphs.
$\Aut$ is the diagonal matrix whose diagonal entry $\Aut(G, G)$, which we usually denote just by $\Aut(G)$, contains the number of automorphisms of the connected hypergraph $G$.

\section{Hypergraphs}
\label{sec:hy}

\subsection{Hypergraph Color Refinement}
\label{sec:hypergraphCR}

Color refinement colors the vertices of a graph $G$ by setting $\ColRef[G]{0}(v) \coloneqq 1$ for every $v \in V(G)$ and $\ColRef[G]{i+1}(v) \coloneqq \MultiSet{\ColRef[G]{i}(u) \mid u \in \Neighbors[G]{v}}$ for every $v \in V(G)$ and every $i \ge 0$.
In a hypergraph, the adjacency of a vertex $v$ is not fully determined by the set of its neighbors, i.e., the set of vertices that share an edge with $v$, as this does not state \textit{how} $v$ is connected to them.
To capture also this information, we rather look at the edges $v$ is incident to:
a coloring of the vertices of a hypergraph induces a coloring of its edges.
For a hypergraph $G$, we define $\HPColRef[G]{0}(v) \coloneqq 1$ for every $v \in V(G)$ and
\begin{align*}
    \HPColRef[G]{i+1}(v) \coloneqq \MultiSet{\MultiSet{\HPColRef[G]{i}(u) \mid u \in \EdgeFunOf{G}(e)} \mid e \in E \text{ with } v \in \EdgeFunOf{G}(e)}
\end{align*}
for every $v \in V(G)$ and every $i \ge 0$.
Color refinement distinguishes two hypergraphs $G$ and $H$ if there is an $i \ge 0$ such that the colorings are \textit{unbalanced}, i.e., that we have $\MultiSet{\HPColRef[G]{i} (v) \mid v \in V(G)} \neq \MultiSet{\HPColRef[H]{i} (v) \mid v \in V(H)}$.

Thus, two vertices of the same color get different colors in a refinement round if they have a different number of incident edges of an induced color $c$.
Note that such an induced color of an edge is a multiset since distinct vertices of the same edge may have the same color.
It is not hard to see that, when interpreting a graph as a hypergraph, the two definitions are equivalent:
an inductive argument yields that excluding the color of $v \in V(G)$ itself from the color $\MultiSet{\HPColRef[G]{i}(u) \mid u \in \EdgeFunOf{G}(e)}$ induced on the edge $e \in E$ with $v \in \EdgeFunOf{G}(e)$ does not make a difference.
Then, the only difference is that each color of a neighbor is placed into its own multiset in the more general definition.

\begin{figure}[htbp]
    \centering
    \begin{minipage}{0.25\textwidth}
    \scalebox{0.8}{
    \begin{tikzpicture}[node distance = \VertexDistance]
		\node[hypervertex, label={270:$v_1$}] (v0) {};
		\node[hypervertex, label={270:$v_2$}, right = of v0] (v1) {};
		\node[hypervertex, label={270:$v_3$}, right = of v1] (v2) {};

		\node[hyperedge, label={90:$e_1$}] at ($(v0)!0.5!(v1) + (0,\EdgeDistance)$) (ge0) {};
		\node[hyperedge, label={90:$e_2$}] at ($(v1)!0.5!(v2) + (0,\EdgeDistance)$) (ge1) {};

		\draw[hyperedgeedge] (v0) to (ge0);
		\draw[hyperedgeedge] (v0) to (ge1);
		\draw[hyperedgeedge] (v1) to (ge0);
		\draw[hyperedgeedge] (v1) to (ge1);
		\draw[hyperedgeedge] (v2) to (ge1);
    \end{tikzpicture}
    }
    \end{minipage}
    \begin{minipage}{0.35\textwidth}
    \begin{tabular}{c|c|c}
        $v$ & $\HPColRef{0}(v)$ & $\HPColRef{1}(v)$\\ \hline
        $v_1$ & $1$ & $\MultiSet{ \MultiSet{1, 1}, \MultiSet{1,1,1}}$\\
        $v_2$ & $1$ & $\MultiSet{ \MultiSet{1, 1}, \MultiSet{1,1,1}}$ \\
        $v_3$ & $1$ & $\MultiSet{ \MultiSet{1,1,1}}$
    \end{tabular}
    \end{minipage}
    \caption{Color refinement on a hypergraph}
    \label{fig:hy:hypergraphCRExample}
\end{figure}
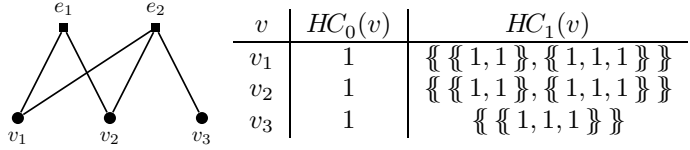

\Cref{fig:hy:hypergraphCRExample} shows an example of color refinement on a hypergraph, which is represented by its incidence graph, where the vertices and edges are depicted as circles and squares, respectively;
this distinction is not made in the incidence graph itself.
To justify our notion of color refinement, we observe its relation to color refinement on the incidence graph of a hypergraph, which also colors its edges:
In a first step, every edge gets assigned the colors of its incident vertices.
In a second step, every vertex gets assigned the colors of its incident edges.
Hence, a single step of color refinement on a hypergraph corresponds to two steps of color refinement on its incidence graph.

However, to formally obtain an equivalence between the two notions, we have to deal with the fact that the additional colors of the edges present in color refinement on an incidence graph may obscure unbalanced vertex partitions, which may happen since an incidence graph does not indicate whether one of its vertices is actually a vertex or an edge of the hypergraph, i.e., vertices of the one hypergraph may be confused with edges of the other.
To avoid this, we differentiate these right from the beginning by defining the \textit{colored incidence graph} $\CIncGraph{G}$ of a hypergraph $G$, which is the vertex-colored graph obtained by taking the incidence graph $\IncGraph{G}$ and coloring the elements of $V(G)$ and $E(G)$ with two different colors, say $1$ for $V(G)$ and $2$ for $E(G)$.
In general, for color refinement on a vertex-colored graph, one has to include a vertex's old color in the new one in every refinement round to guarantee that we indeed obtain a refinement.
However, a simple inductive argument yields that this is not necessary for colored incidence graphs.

\begin{lemma}
    \label{le:hy:HyperCRIncedenceCR}
    For all hypergraphs $G$ and $H$, the following are equivalent:
    \begin{enumerate}
        \item Color refinement does not distinguish $G$ and $H$. \label{le:hy:HyperCRIncedenceCR:Multi}
        \item Color refinement does not distinguish $\CIncGraph{G}$ and $\CIncGraph{H}$. \label{le:hy:HyperCRIncedenceCR:Inc}
    \end{enumerate}
\end{lemma}
\begin{proof}
    First, we prove that we have $\HPColRef[G]{i} = \restr{\ColRef[\CIncGraph{G}]{2i}}{V(G)}$ for every hypergraph $G$ and every $i \ge 0$, which we prove by induction on $i$.
    The induction basis $i = 0$ is trivial by our choice of colors of $\CIncGraph{G}$.
    For the inductive step $i > 0$, we have
    \begin{align*}
        \ColRef[\CIncGraph{G}]{2i}(v) &= \MultiSet{\ColRef[\CIncGraph{G}]{2i - 1}(u) \mid u \in \Neighbors[\CIncGraph{G}]{v}}\\
        &= \MultiSet{\ColRef[\CIncGraph{G}]{2i - 1}(e) \mid e \in E(G) \text{ with } v \in \EdgeFunOf{G}(e)}\\
        &= \MultiSet{\MultiSet{\ColRef[\CIncGraph{G}]{2(i-1)}(u) \mid u \in \Neighbors[\CIncGraph{G}]{e}} \mid e \in E(G) \text{ with } v \in \EdgeFunOf{G}(e)}\\
        &= \MultiSet{\MultiSet{\ColRef[\CIncGraph{G}]{2(i-1)}(u) \mid u \in \EdgeFunOf{G}(e)} \mid e \in E(G) \text{ with } v \in \EdgeFunOf{G}(e)}\\
        &=\HPColRef[G]{i}(v)
    \end{align*}
    for every $v \in V(G)$, where we obtain the last equality by applying the induction hypothesis.

    To prove the equivalence of \ref{le:hy:HyperCRIncedenceCR:Multi} and \ref{le:hy:HyperCRIncedenceCR:Inc}, let $G$ and $H$ be hypergraphs.
    The backward direction \Direction{\ref{le:hy:HyperCRIncedenceCR:Inc}}{\ref{le:hy:HyperCRIncedenceCR:Multi}} follows directly:
    if we have that
    \begin{align*}
        \MultiSet{\ColRef[\CIncGraph{G}]{2i}(v) \mid v \in V(\CIncGraph{G})} = \MultiSet{\ColRef[\CIncGraph{H}]{2i}(v) \mid v \in V(\CIncGraph{H})}
    \end{align*}
    holds for every $i \ge 0$, then the restrictions of these colorings to $V(G)$ and $V(H)$ are also balanced due to the initial colorings chosen for the colored incidence graphs, and by the identity proven above, these restrictions are just the colorings computed by color refinement on the hypergraphs $G$ and $H$.
    This is actually the only point of the proof where the coloring of a colored incidence graph is needed.

    Proving the forward direction \Direction{\ref{le:hy:HyperCRIncedenceCR:Multi}}{\ref{le:hy:HyperCRIncedenceCR:Inc}} is more involved since we have to make sure that the colors of the edges in the incidence graphs do not contain any additional information. Intuitively, this holds since every edge is connected to a vertex, which receives its information in the next refinement round.
    For a formal proof, assume that the two colorings $\HPColRef[G]{i}$ and $\HPColRef[H]{i}$ are balanced for every $i \ge 0$, and to prove the claim, choose $i \ge 0$ large enough such that the colorings $\ColRef[\CIncGraph{G}]{2i}$ and $\ColRef[\CIncGraph{H}]{2i}$ are stable.
    By the assumption and the identity proven above, these colorings are balanced on $V(G)$ and $V(H)$, i.e., we have
    \begin{align*}
        \MultiSet{\ColRef[\CIncGraph{G}]{2i}(v) \mid v \in V(G)} = \MultiSet{\ColRef[\CIncGraph{H}]{2i}(v) \mid v \in V(H)},
    \end{align*}
    and it suffices to prove that they are also balanced on $E(G)$ and $E(H)$.
    Then, we have that the whole colorings are balanced, and since they are stable, the colorings in all further refinement rounds are also balanced.

    Let $c \in C \coloneqq \ImageOf{\ColRef[\CIncGraph{G}]{2i}} \cup \ImageOf{\ColRef[\CIncGraph{H}]{2i}}$ be a color used by one of the colorings.
    Since the colorings are stable and the first round of color refinement determines the degrees of the vertices, every $e \in E(G) \cup E(H)$ of color $c$ has the same size~$m$.
    Using $\MultiSet{\ColRef[\CIncGraph{G}]{2i}(v) \mid v \in V(G)} = \MultiSet{\ColRef[\CIncGraph{H}]{2i}(v) \mid v \in V(H)}$ and the stability of the colorings yields
    \begin{align*}
        \abs{\Set{e \in E(G) \mid \ColRef[\CIncGraph{G}]{2i}(e) = c}} &= \frac{1}{m} \cdot \sum_{v \in V(G)} \abs{\Neighbors[\CIncGraph{G}]{v} \cap {\ColRef[\CIncGraph{G}]{2i}}^{-1}(c)}\\
        &= \frac{1}{m} \cdot \sum_{d \in C} \sum_{\substack{v \in V(G),\\ \ColRef[\CIncGraph{G}]{2i}(v) = d}} \abs{\Neighbors[\CIncGraph{G}]{v} \cap {\ColRef[\CIncGraph{G}]{2i}}^{-1}(c)}\\
        &= \frac{1}{m} \cdot \sum_{d \in C} \sum_{\substack{v \in V(H),\\ \ColRef[\CIncGraph{H}]{2i}(v) = d}} \abs{\Neighbors[\CIncGraph{H}]{v} \cap {\ColRef[\CIncGraph{H}]{2i}}^{-1}(c)}\\
        &= \abs{\Set{e \in E(H) \mid \ColRef[\CIncGraph{H}]{2i}(e) = c}}
    \end{align*}
    since every edge of a hypergraph is non-empty.
    This means that the colorings are balanced on $E(G)$ and $E(H)$ and finishes the proof.
\end{proof}

\subsection{Incidence Homomorphisms}
\label{sec:homomorphismsFromBACHs}

Recall that a hypergraph is connected and Berge-acyclic if and only if its incidence graph is a tree.
With the Colored Tree Theorem, \Cref{le:hy:HyperCRIncedenceCR} already yields that two hypergraphs $G$ and $H$ are not distinguished by color refinement if and only if $\HOM{\CTs}{\CIncGraph{G}} = \HOM{\CTs}{\CIncGraph{H}}$, i.e., we already have a characterization of color refinement by counting homomorphisms from vertex-colored trees to the hypergraphs' incidence graphs.
This motivates a \enquote{reduction} to prove \Cref{th:hy:countingBACHHoms}, i.e., instead of adapting the proof of the Tree Theorem by defining an unfolding of a hypergraph into a Berge-acyclic one, we relate homomorphisms between colored incidence graphs back to homomorphisms between hypergraphs.

To this end, we first re-formulate $\HOM{\CTs}{\CIncGraph{G}} = \HOM{\CTs}{\CIncGraph{H}}$ in hypergraph terms.
Observe that, at this point, it is convenient that we consider hypergraphs with parallel edges because, when interpreting a colored tree as an incidence graph of a hypergraph, it may very well have parallel edges, or more precisely, parallel loops.
Thus, when taking the step from vertex-colored trees to hypergraphs, the only noteworthy special case is the colored tree corresponding to an empty edge, which does not have a corresponding hypergraph as empty edges are disallowed by definition.

However, just interpreting vertex-colored trees as hypergraphs does not suffice as, for hypergraphs $G$ and $H$, the homomorphisms between the colored incidence graphs $\CIncGraph{G}$ and $\CIncGraph{H}$ do not necessarily correspond to homomorphisms between $G$ and $H$.
While every homomorphism $(h_V, h_E)$ from $G$ to $H$ gives us a corresponding homomorphism $h_V \cup h_E$ from $\CIncGraph{G}$ to $\CIncGraph{H}$, the converse does not hold:
a homomorphism from $\CIncGraph{G}$ to $\CIncGraph{H}$ does not have to map the vertices of an edge of $G$ to a \textit{full} edge of $H$ but only to a subset of such an edge, cf.\ \Cref{fig:hy:mergingExample}.
To capture this behavior in terms of hypergraphs, for hypergraphs $G$ and $H$, we call a pair $(h_V, h_E)$ of mappings $h_V \colon V(G) \to V(H)$ and $h_E \colon E(G) \to E(H)$ satisfying $h_V(\EdgeFunOf{G}(e)) \subseteq \EdgeFunOf{H}(h_E(e))$ for every $e \in E(G)$ an \textit{incidence homomorphism} from $G$ to $H$.
That is, the equality in the definition of a homomorphism is relaxed to an inclusion, which also means that every homomorphism \textit{is} an incidence homomorphism.

Observe that we have a one-to-one correspondence between the incidence homomorphisms from $G$ to $H$ and the homomorphisms from $\CIncGraph{G}$ to $\CIncGraph{H}$.
In particular, if we let $\IHom(G,H)$ denote the number of incidence homomorphisms from $G$ to $H$, we have $\IHom(G, H) = \Hom(\CIncGraph{G}, \CIncGraph{H})$.
This lets us express the requirement $\HOM{\CTs}{\CIncGraph{G}} = \HOM{\CTs}{\CIncGraph{H}}$ in terms of connected Berge-acyclic hypergraphs, where a simple interpolation argument takes care of the colored tree corresponding to an empty edge.

\begin{lemma}
    \label{th:hy:countingMBACHIncHoms}
    For all hypergraphs $G$ and $H$, the following are equivalent:
    \begin{enumerate}
        \item $\IHOM{\BACHs}{G} = \IHOM{\BACHs}{H}$. \label{th:hy:countingMBACHIncHoms:BACHs}
        \item Color refinement does not distinguish $G$ and $H$. \label{th:hy:countingMBACHIncHoms:ColorRefinement}
    \end{enumerate}
\end{lemma}
\begin{proof}
    By \Cref{le:hy:HyperCRIncedenceCR} and the Colored Tree Theorem, it suffices to prove the equivalence of \ref{th:hy:countingMBACHIncHoms:BACHs} and $\HOM{\CTs}{\CIncGraph{G}} = \HOM{\CTs}{\CIncGraph{H}}$ for all hypergraphs $G$ and $H$, where the backward direction directly follows because, for a connected Berge-acyclic hypergraph $B$, the incidence graph $\CIncGraph{B}$ is a tree, which yields $\IHom(B, G) = \Hom(\CIncGraph{B}, \CIncGraph{G}) = \Hom(\CIncGraph{B}, \CIncGraph{H}) = \IHom(B, H)$.

    For the forward direction, we let $G$ be a hypergraph and prove that the entries of $\HOM{\CTs}{\CIncGraph{G}}$ are determined by those of $\IHOM{\BACHs}{G}$.
    This direction is not as straightforward as the backward direction because not every vertex-colored tree occurs as the colored incidence graph of a hypergraph.
    If a vertex-colored tree $T$ is isomorphic to the colored incidence graph of a connected Berge-acyclic hypergraph~$B$, then we have $\Hom(T, \CIncGraph{G}) = \IHom(B, G)$ and are done.
    Otherwise, we distinguish two cases:
    If $T$ uses a color distinct from the ones used in the definition of the colored incidence graph or if there is an edge between vertices of the same color, then there is no homomorphism from $T$ to $\CIncGraph{G}$, and we trivially have $\Hom(T, \CIncGraph{G}) = 0$.
    For the second case, we have to assume that $T$ is an isolated vertex of the color used for hyperedges.
    In this case, we have $\Hom(T, \CIncGraph{G}) = \abs{E(G)}$ and use an interpolation argument to show that the number of edges $\abs{E(G)}$ of $G$ can be obtained from the entries of $\IHOM{\BACHs}{G}$.

%
%

    \newcommand{\BkOf}[1]{B_{#1}}%
    \newcommand{\AbsEGiOf}[1]{\abs{E(G)}_{#1}}%
    Let $n \coloneqq \abs{V(G)} = \IHom(\KnOf{1},G)$ be the number of vertices of $G$ and, for every $i \in \NumbersTo{n}$, let
    $\AbsEGiOf{i} \coloneqq \abs{\Set{e \in E(G) \mid \abs{\EdgeFunOf{G}(e)} = i}}$ denote the number of edges of $G$ of size $i$.
    Then, we have $\abs{E(G)} = \sum_{i = 1}^{n} \AbsEGiOf{i}$, and it suffices to show that these values can be obtained from $\IHOM{\BACHs}{G}$.
    For $k \ge 1$, we define $\BkOf{k}$ by setting $V(\BkOf{k}) \coloneqq \Set{v_1, \dots, v_k}$, $E(\BkOf{k}) \coloneqq \Set{e}$, and $\EdgeFunOf{\BkOf{k}}(e) \coloneqq \Set{v_1, \dots, v_k}$, where $v_1, \dots, v_k$ are fresh vertices and $e$ is a fresh edge, i.e., $\BkOf{k}$ is a connected Berge-acyclic hypergraph with a single edge that connects $k$ vertices.

    Observe that we have $\IHom(\BkOf{k}, G) = \sum_{i=1}^{n} i^k \cdot \AbsEGiOf{i}$ for every $k \ge 1$, which yields the system
    \begin{align*}
        \begin{pmatrix}
            1^1 & 2^1 & \hdots & n^1\\
            1^2 & 2^2 & \hdots & n^2\\
            \vdots & \vdots & \ddots & \vdots\\
            1^n & 2^n & \hdots & n^n \\
        \end{pmatrix}
        \cdot
        \begin{pmatrix}
            \AbsEGiOf{1}\\
            \AbsEGiOf{2}\\
            \vdots\\
            \AbsEGiOf{n}
        \end{pmatrix}
        =
        \begin{pmatrix}
            \IHom(\BkOf{1}, G)\\
            \IHom(\BkOf{2}, G)\\
            \vdots\\
            \IHom(\BkOf{n}, G)
        \end{pmatrix}
    \end{align*}
    of linear equations, where the matrix is invertible:
    the Vandermonde matrix $V(0, \dots, n)$ is invertible since the values $0, \dots, n$ are pairwise distinct, and if we delete the first row and column of $V(0, \dots, n)$, which does not affect its determinant by the Laplace expansion, we obtain the transpose of the matrix above.
    Hence, we get that the values $\AbsEGiOf{1}, \dots, \AbsEGiOf{n}$ can be obtained from $\IHOM{\BACHs}{G}$.
\end{proof}

\subsection{Homomorphisms from Berge-Acyclic Hypergraphs}
\label{sec:hy:BACHs}

With \Cref{th:hy:countingMBACHIncHoms}, it remains to show that counting incidence homomorphisms from $\BACHs$ is equivalent to counting homomorphisms from $\BACHs$.
To this end, we call an incidence homomorphism $(h_V, h_E)$ from a hypergraph $G$ to a hypergraph $H$ \textit{locally injective}, \textit{locally surjective}, and \textit{locally bijective} if, for every $e \in E(G)$, the restriction $\restr{h_V}{\EdgeFunOf{G}(e)} \colon \EdgeFunOf{G}(e) \to \EdgeFunOf{H}(h_E(e))$ of $h_V$ to the vertices of $e$ is injective, surjective, and bijective, respectively.
For a connected hypergraph $G$ and a hypergraph $H$, we denote the number of locally injective incidence homomorphisms by $\LocInjIHom(G, H)$ and, since an incidence homomorphism is locally surjective if and only if it is a homomorphism, the number of locally bijective incidence homomorphisms by $\LocInjHom(G, H)$.

The main work is spread across three lemmas:
Together, \Cref{le:hy:IHomLocInjIHom} and \Cref{le:hy:LocInjIHomLocInjHom} \enquote{balance} incidence homomorphisms to locally bijective incidence homomorphisms by first relating incidence homomorphisms to locally injective incidence homomorphisms and then, from there on, to locally bijective incidence homomorphisms.
Analogously to \Cref{le:hy:IHomLocInjIHom}, \Cref{le:hy:HomEInjHom} relates homomorphisms to locally injective homomorphisms or, in other words, locally bijective incidence homomorphisms.

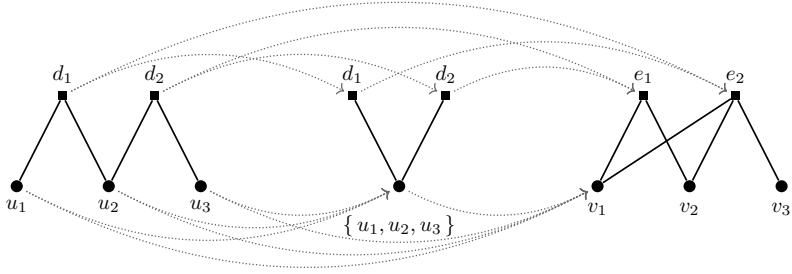
\begin{figure}[tbp]
	\centering
    \scalebox{0.8}{
	\begin{tikzpicture}[node distance = \VertexDistance cm]
		\node[hypervertex, label={270:$u_1$}] (t0) {};
		\node[hypervertex, label={270:$u_2$}, right = of t0] (t1) {};
		\node[hypervertex, label={270:$u_3$}, right = of t1] (t2) {};

		\node[hyperedge, label={90:$d_1$}] at ($(t0)!0.5!(t1) + (0,\EdgeDistance)$) (e0) {};
		\node[hyperedge, label={90:$d_2$}] at ($(t1)!0.5!(t2) + (0,\EdgeDistance)$) (e1) {};

		\draw[hyperedgeedge] (t0) to (e0);
		\draw[hyperedgeedge] (t1) to (e0);
		\draw[hyperedgeedge] (t1) to (e1);
		\draw[hyperedgeedge] (t2) to (e1);

		\node[right = 1.5cm of t2] (b0) {};
		\node[hypervertex, label={[label distance = 0.25cm]270:$\Set{u_1, u_2, u_3}$}, right = of b0] (b1) {};
		\node[right = of b1] (b2) {};
		\node[right = of b2] (b3) {};

		\node[hyperedge, label={90:$d_1$}] at ($(b0)!0.5!(b1) + (0,\EdgeDistance)$) (be0) {};
		\node[hyperedge, label={90:$d_2$}] at ($(b1)!0.5!(b2) + (0,\EdgeDistance)$) (be1) {};

		\draw[hyperedgeedge] (b1) to (be0);
		\draw[hyperedgeedge] (b1) to (be1);

		\node[hypervertex, label={[label distance=0.04cm]270:$v_1$}, right = 1.5cm of b2] (v0) {};
		\node[hypervertex, label={[label distance=0.04cm]270:$v_2$}, right = of v0] (v1) {};
		\node[hypervertex, label={[label distance=0.04cm]270:$v_3$}, right = of v1] (v2) {};

		\node[hyperedge, label={90:$e_1$}] at ($(v0)!0.5!(v1) + (0,\EdgeDistance)$) (ge0) {};
		\node[hyperedge, label={90:$e_2$}] at ($(v1)!0.5!(v2) + (0,\EdgeDistance)$) (ge1) {};

		\draw[hyperedgeedge] (v0) to (ge0);
		\draw[hyperedgeedge] (v0) to (ge1);
		\draw[hyperedgeedge] (v1) to (ge0);
		\draw[hyperedgeedge] (v1) to (ge1);
		\draw[hyperedgeedge] (v2) to (ge1);

		\draw[overlay, mappingHP, bend left = \TadLongerShorterBend] (e0) to (ge1);
		\draw[overlay, mappingHP, bend left = \TadLongerShorterBend] (e1) to (ge0);

		\draw[overlay, mappingHP, bend right = \TadLongerShorterBend] (t0) to (v0);
		\draw[overlay, mappingHP, bend right = \TadLongerShorterBend] (t1) to (v0);
		\draw[overlay, mappingHP, bend right = \TadLongerShorterBend] (t2) to (v0);

		\draw[mappingHP, bend left = \TadLongerShorterBend] (e0) to (be0);
		\draw[mappingHP, bend left = \TadLongerShorterBend] (e1) to (be1);

		\draw[mappingHP, bend right = \TadLongerShorterBend] (t0) to (b1);
		\draw[mappingHP, bend right = \TadLongerShorterBend] (t1) to (b1);
		\draw[mappingHP, bend right = \TadLongerShorterBend] (t2) to (b1);

		\draw[mappingHP, bend left = \TadLongerShorterBend] (be0) to (ge1);
		\draw[mappingHP, bend left = \TadLongerShorterBend] (be1) to (ge0);

		\draw[mappingHP, bend right = \TadLongerShorterBend] (b1) to (v0);
	\end{tikzpicture}
    }
	\caption{Decomposition of an incidence homomorphism into a locally merging homomorphism and a locally injective incidence homomorphism}
	\label{fig:hy:mergingExample}
\end{figure}

While our goal is to relate incidence homomorphisms to locally surjective incidence homomorphisms, we are forced to take the detour that is local injectivity due to the way we prove \Cref{le:hy:LocInjIHomLocInjHom}:
We fill up edges that are mapped non-surjectively by adding leaves, i.e., vertices that are part of exactly one edge.
Without this injectivity, which we achieve by merging vertices within an edge that are mapped to the same vertex, these added leaves may be mapped to the same vertex again causing us to overcount endlessly.
With local injectivity, also achieving local surjectivity is possible as a locally bijective incidence homomorphism has to map an edge to an edge of exactly the same size.
Thus, if we use leaves to fill up an edge to the size of the target edges, we do not overcount as we do not count incidence homomorphisms where adding fewer leaves would have sufficed.
Note that, in our setting, it is crucial that we only consider such a local form of injectivity;
we have to make sure the Berge-acyclicity is preserved when merging vertices.

To relate incidence homomorphisms to locally injective incidence homomorphisms, we define \textit{locally merging} homomorphisms, which only allow vertices to be mapped to the same vertex if they are part of the same edge.
To this end, we first define the relation $\MergeEquivSpaceOf{h_V} \subseteq V(G) \times V(G)$ for an incidence homomorphism $(h_V, h_E)$ between two hypergraphs $G$ and $H$ by letting $u \MergeEquivOf{h_V} v$ if there is a walk $v_0, e_1, \dots, v_k$ from $u$ to $v$ in $G$ with $h_V(v_{i-1}) = h_V(v_{i})$ for every $i \in \NumbersTo{k}$.
Clearly, $\MergeEquivOf{h_V}$ is an equivalence relation, and for all $u,v \in V(G)$, we have that $u \MergeEquivOf{h_V} v$ implies $h_V(u) = h_V(v)$.
We call a homomorphism $(h_V, h_E)$ between hypergraphs $G$ and $H$ locally merging if
\begin{enumerate}
    \item $h_V(u) = h_V(v)$ if and only if $u \MergeEquivOf{h_V} v$ for all $u,v \in V(G)$,
    \item $h_V$ is surjective, and \label{en:hy:merging:verticesSurjective}
    \item $h_E$ is bijective, \label{en:hy:merging:edgesBijective}
\end{enumerate}
and, for connected hypergraphs $G$ and $H$, we let $\EMergHom(G, H)$ be the number of such homomorphisms from $G$ to $H$.

By decomposing incidence homomorphisms into locally merging homomorphisms and locally injective incidence homomorphisms as in \Cref{fig:hy:mergingExample}, we obtain \Cref{le:hy:IHomLocInjIHom}.
The crucial argument is the fact that the intermediate hypergraph is uniquely determined by $(h_V, h_E)$, i.e., every decomposition of $(h_V, h_E)$ has to use the same intermediate hypergraph.
Note that, by merging vertices to obtain the intermediate hypergraph, parallel loops may be created even when decomposing an incidence homomorphism between simple hypergraphs.
Moreover, these parallel loops may have to be mapped to different edges, making it impossible to merge them into a single loop.
Since, for such a decomposition, automorphisms of the intermediate hypergraph can be used to obtain a different decomposition, we have to divide by the number of automorphisms.
Note that the identity of \Cref{le:hy:IHomLocInjIHom} is stated for arbitrary connected hypergraphs;
once it is needed, we restrict it to Berge-acyclic ones.

\newcommand{\IncV}{h_V}
\newcommand{\IncE}{h_E}
\newcommand{\Inc}{(\IncV, \IncE)}

\newcommand{\MerV}[1][]{{g_V^{#1}}}
\newcommand{\MerE}[1][]{{g_E^{#1}}}
\newcommand{\Mer}[1][]{(\MerV[#1], \MerE[#1])}
\newcommand{\MerVPrime}{{g'_V}}
\newcommand{\MerEPrime}{{g'_E}}
\newcommand{\MerPrime}{(\MerVPrime, \MerEPrime)}

\newcommand{\LocInjIncV}[1][]{{h_V^{#1}}}
\newcommand{\LocInjIncE}[1][]{{h_E^{#1}}}
\newcommand{\LocInjInc}[1][]{(\LocInjIncV[#1], \LocInjIncE[#1])}
\newcommand{\LocInjIncVPrime}{{h'_V}}
\newcommand{\LocInjIncEPrime}{{h'_E}}
\newcommand{\LocInjIncPrime}{(\LocInjIncVPrime, \LocInjIncEPrime)}

\newcommand{\MerLocInjInc}[1][]{(\Mer[#1], \LocInjInc[#1])}
\newcommand{\MerLocInjIncPrime}{(\MerPrime, \LocInjIncPrime)}
\newcommand{\MerLocInjIncFam}{(\MerLocInjInc[\pi])_{\pi \in \AutGroup(G')}}

\newcommand{\Merged}{G^{m}}
\newcommand{\IsoV}{{\sigma_{V}^m}}
\newcommand{\IsoE}{{\sigma_{E}^m}}
\newcommand{\Iso}{(\IsoV, \IsoE)}
\begin{lemma}
    \label{le:hy:IHomLocInjIHom}
    We have $\IHom = \EMergHom \cdot \Aut^{-1} \cdot \LocInjIHom$.
    The matrix $\EMergHom$ is invertible and lower triangular.
\end{lemma}

For the special case of homomorphisms, i.e., locally surjective incidence homomorphisms, the proof of \Cref{le:hy:IHomLocInjIHom} also directly yields \Cref{le:hy:HomEInjHom}.

\begin{lemma}
    \label{le:hy:HomEInjHom}
    We have $\Hom = \EMergHom \cdot \Aut^{-1} \cdot \LocInjHom$.
\end{lemma}

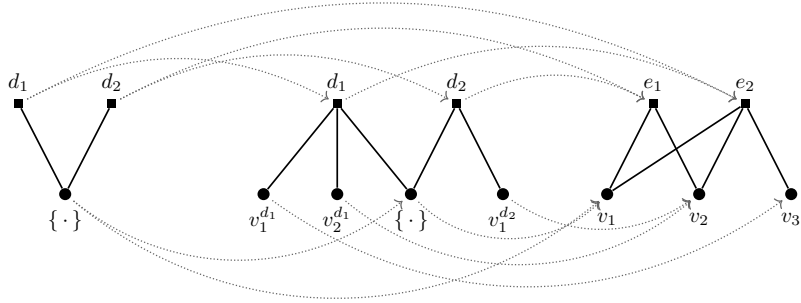
\begin{figure}[htbp]
	\centering
    \scalebox{0.8}{
	\begin{tikzpicture}[node distance = \VertexDistance cm]
		\node[overlay] (t0) {};
		\node[hypervertex, label={[label distance=0.03cm]270:$\Set{\cdot}$}, right = of t0] (t1) {};
		\node[right = of t1] (t2) {};

		\node[hyperedge, label={90:$d_1$}] at ($(t0)!0.5!(t1) + (0,\EdgeDistance)$) (e0) {};
		\node[hyperedge, label={90:$d_2$}] at ($(t1)!0.5!(t2) + (0,\EdgeDistance)$) (e1) {};

		\draw[hyperedgeedge] (t1) to (e0);
		\draw[hyperedgeedge] (t1) to (e1);

		\node[hypervertex, label={[label distance=-0.03cm]270:$v^{d_1}_1$}, right = 1.5cm of t2] (b0) {};
		\node[hypervertex, label={[label distance=0.03cm]270:$\Set{\cdot}$}, right = 1.7 * \VertexDistance of b0] (b1) {};
		\node[hypervertex, label={[label distance=-0.03cm]270:$v^{d_2}_1$}, right = of b1] (b2) {};
		\node[hypervertex, label={[label distance=-0.03cm]270:$v^{d_1}_2$}] (b3) at ($(b0)!0.5!(b1)$) {};

		\node[hyperedge, label={90:$d_1$}] at ($(b0)!0.5!(b1) + (0,\EdgeDistance)$) (be0) {};
		\node[hyperedge, label={90:$d_2$}] at ($(b1)!0.5!(b2) + (0,\EdgeDistance)$) (be1) {};

		\draw[hyperedgeedge] (b0) to (be0);
		\draw[hyperedgeedge] (b1) to (be0);
		\draw[hyperedgeedge] (b1) to (be1);
		\draw[hyperedgeedge] (b2) to (be1);
		\draw[hyperedgeedge] (b3) to (be0);

		\node[hypervertex, label={[label distance=0.08cm]270:$v_1$}, right = 1.5cm of b2] (v0) {};
		\node[hypervertex, label={[label distance=0.08cm]270:$v_2$}, right = of v0] (v1) {};
		\node[hypervertex, label={[label distance=0.08cm]270:$v_3$}, right = of v1] (v2) {};

		\node[hyperedge, label={90:$e_1$}] at ($(v0)!0.5!(v1) + (0,\EdgeDistance)$) (ge0) {};
		\node[hyperedge, label={90:$e_2$}] at ($(v1)!0.5!(v2) + (0,\EdgeDistance)$) (ge1) {};

		\draw[hyperedgeedge] (v0) to (ge0);
		\draw[hyperedgeedge] (v0) to (ge1);
		\draw[hyperedgeedge] (v1) to (ge0);
		\draw[hyperedgeedge] (v1) to (ge1);
		\draw[hyperedgeedge] (v2) to (ge1);

		\draw[overlay, mappingHP, bend left = \TadLongerShorterBend] (e0) to (ge1);
		\draw[overlay, mappingHP, bend left = \TadLongerShorterBend] (e1) to (ge0);

		\draw[overlay, mappingHP, bend right = \NotSoLongBend] (t1) to (v0);

		\draw[mappingHP, bend left = \TadLongerShorterBend] (e0) to (be0);
		\draw[mappingHP, bend left = \TadLongerShorterBend] (e1) to (be1);

		\draw[mappingHP, bend right = 39] (t1) to (b1);

		\draw[mappingHP, bend left = \TadLongerShorterBend] (be0) to (ge1);
		\draw[mappingHP, bend left = \TadLongerShorterBend] (be1) to (ge0);

		\draw[overlay, mappingHP, bend right = 36] (b0) to (v2);
		\draw[mappingHP, bend right = 46] (b1) to (v0);
		\draw[mappingHP, bend right = 35] (b2) to (v1);
		\draw[mappingHP, bend right = \NotSoLongBend] (b3) to (v1);
	\end{tikzpicture}
    }
	\caption{Decomposition of a locally injective incidence homomorphism into a leaf-adding incidence homomorphism and a locally injective homomorphism}
	\label{fig:hy:leafAddingExample}
\end{figure}

To prove \Cref{le:hy:LocInjIHomLocInjHom}, we define \textit{leaf-adding} incidence homomorphisms, which are embeddings of a hypergraph into another one that has no additional vertices or edges with the exception of leaves.
For this, we need the notion of a \textit{strong} incidence homomorphism between hypergraphs $G$ and $H$, which is an incidence homomorphism $(h_V, h_E)$ from $G$ to $H$ that additionally satisfies the inclusion $h_V^{-1}(\EdgeFunOf{H}(h_E(e)) \cap \ImageOf{h_V}) \subseteq \EdgeFunOf{G}(e)$ for every $e \in E(G)$;
it is actually not hard to see that this is equivalent to requiring that the corresponding homomorphism between the colored incidence graphs $\CIncGraph{G}$ and $\CIncGraph{H}$ is a strong homomorphism.
We call an incidence homomorphism $(h_V, h_E)$ between hypergraphs $G$ and $H$ leaf-adding if
\begin{enumerate}
    \item $(h_V, h_E)$ is a strong incidence homomorphism, \label{en:hy:leafadding:strong}
    \item $h_V$ is injective, \label{en:hy:leafadding:injective}
    \item $h_E$ is bijective, and \label{en:hy:leafadding:edgesBijective}
    \item the vertices $V(H) \setminus \ImageOf{h_V}$ are leaves of $H$, \label{en:hy:leafadding:notHitLeaf}
\end{enumerate}
\noindent and, for connected hypergraphs $G$ and $H$, we let $\LeafAddingI(G, H)$ be the number of leaf-adding incidence homomorphisms from $G$ to $H$.
Similarly to the proof of \Cref{le:hy:IHomLocInjIHom}, the proof of \Cref{le:hy:LocInjIHomLocInjHom} decomposes locally injective incidence homomorphisms into leaf-adding incidence homomorphisms and locally injective homomorphisms as in \Cref{fig:hy:leafAddingExample}.
Again, this identity is proven for arbitrary connected hypergraphs, and we restrict it to Berge-acyclic ones once it is needed.

\newcommand{\LAddV}[1][]{{g_V^{#1}}}
\newcommand{\LAddE}[1][]{{g_E^{#1}}}
\newcommand{\LAdd}[1][]{(\LAddV[#1], \LAddE[#1])}
\newcommand{\LAddVPrime}{{g'_V}}
\newcommand{\LAddEPrime}{{g'_E}}
\newcommand{\LAddPrime}{(\LAddVPrime, \LAddEPrime)}

\newcommand{\LocInjV}[1][]{{h_V^{#1}}}
\newcommand{\LocInjE}[1][]{{h_E^{#1}}}
\newcommand{\LocInj}[1][]{(\LocInjV[#1], \LocInjE[#1])}
\newcommand{\LocInjVPrime}{{h'_V}}
\newcommand{\LocInjEPrime}{{h'_E}}
\newcommand{\LocInjPrime}{(\LocInjVPrime, \LocInjEPrime)}

\newcommand{\LAddLocInj}[1][]{(\LAdd[#1], \LocInj[#1])}
\newcommand{\LAddLocInjPrime}{(\LAddPrime, \LocInjPrime)}
\newcommand{\LAddLocInjFam}{(\LAddLocInj[\pi])_{\pi \in \AutGroup(G')}}

\newcommand{\Balanced}{G^{a}}
\newcommand{\IsoBalV}{{\sigma_{V}^a}}
\newcommand{\IsoBalE}{{\sigma_{E}^a}}
\newcommand{\IsoBal}{(\IsoBalV, \IsoBalE)}
\begin{lemma}
    \label{le:hy:LocInjIHomLocInjHom}
    We have $\LocInjIHom = \LeafAddingI \cdot \Aut^{-1} \cdot \LocInjHom$.
    The matrix $\LeafAddingI$ is invertible and upper triangular.
\end{lemma}

We have all we need to prove that counting incidence homomorphisms from $\BACHs$ is equivalent to counting homomorphisms from $\BACHs$.
Combining \Cref{le:hy:IHomLocInjIHom} and \Cref{le:hy:LocInjIHomLocInjHom} yields $\IHom = \EMergHom \cdot \Aut^{-1} \cdot \LeafAddingI \cdot \Aut^{-1} \cdot \LocInjHom$,
and \Cref{le:hy:HomEInjHom} states that we have $\Hom = \EMergHom \cdot \Aut^{-1} \cdot \LocInjHom$.
Even with the invertibility of $\EMergHom$ and $\LeafAddingI$, the proof of \Cref{le:hy:incidencehomshoms} is not trivial as the inverse of the upper triangular matrix $\LeafAddingI$ is still an upper triangular matrix, and hence, left multiplication with it may be undefined.
This, however, can be avoided by considering finite submatrices as in \cite{Dell2018}.
This proof finishes our \enquote{reduction} and, hence, the proof of \Cref{th:hy:countingBACHHoms} as it follows immediately from \Cref{th:hy:countingMBACHIncHoms} and \Cref{le:hy:incidencehomshoms}.

\begin{lemma}
    \label{le:hy:incidencehomshoms}
    For all hypergraphs $G$ and $H$, the following are equivalent:
    \begin{enumerate}
        \item $\IHOM{\BACHs}{G} = \IHOM{\BACHs}{H}$. \label{le:hy:incidencehomshoms:IncHom}
        \item $\HOM{\BACHs}{G} = \HOM{\BACHs}{H}$. \label{le:hy:incidencehomshoms:Hom}
    \end{enumerate}
\end{lemma}
\begin{proof}
    We prove that $\LOCINJIHOM{\BACHs}{G} = \LOCINJIHOM{\BACHs}{H}$ holds if and only if $\LOCINJHOM{\BACHs}{G} = \LOCINJHOM{\BACHs}{H}$;
    the remaining two equivalences between counting (incidence) homomorphisms and counting locally injective (incidence) homomorphisms are easier to prove and follow in a similar fashion from \Cref{le:hy:IHomLocInjIHom} and \Cref{le:hy:HomEInjHom}, respectively.

    We show that, for every hypergraph~$G$, the entries of $\LOCINJIHOM{\BACHs}{G}$ are determined by the entries of $\LOCINJHOM{\BACHs}{G}$ and vice versa.
    For a connected Berge-acyclic hypergraph $B$ and a connected hypergraph $G$, having $\LeafAddingI(B, G) > 0$ implies that $G$ is also Berge-acyclic since $G$ results from $B$ by adding leaves
    Hence, we are able to restrict the class of all connected hypergraphs in the identity of \Cref{le:hy:LocInjIHomLocInjHom} to the class $\BACHs$, which gives us $\LOCINJIHOM{\BACHs}{G} = \restr{\LeafAddingI}{\BACHs \times \BACHs} \cdot \restr{\Aut}{\BACHs \times \BACHs}^{-1} \cdot \LOCINJHOM{\BACHs}{G}$;
    this already yields the backward direction.

    $\restr{\EMergHom}{\BACHs \times \BACHs}$ is still an invertible upper triangular matrix, and so is its inverse $\restr{\EMergHom}{\BACHs \times \BACHs}^{-1}$, which means that left multiplication with it may not be defined.
    However, we are able to circumvent this by considering finite submatrices.
    \newcommand{\BACHsmn}{\mathcal{B\!A}^{m}_{n}}
    To this end, for $m \ge 0$, let $\BACHsmn \subseteq \BACHs$ be the class of all connected Berge-acyclic hypergraphs with exactly $m$ edges and at most $n \coloneqq \abs{V(G)} = \LocInjIHom(\KnOf{1}, G)$ vertices in each edge, which clearly is finite.
    For every $B \in \BACHsmn$, the sum $\sum_{B' \in \BACHs} \LeafAddingI(B,B') \cdot \Aut^{-1}(B') \cdot \LocInjHom(B', G)$ can be restricted to hypergraphs $B' \in \BACHsmn$ since $\LeafAddingI(B,B') > 0$ implies that $B'$ has exactly as many edges as $B$, i.e., $m$ edges, and $\LocInjHom(B', G) > 0$ implies that every edge of $B'$ contains at most $n$ vertices.
    We obtain that $\LOCINJIHOM{\BACHsmn}{G}$ is equal to
    \begin{equation*}
        \restr{\LeafAddingI}{\BACHsmn \times \BACHsmn} \cdot \restr{\Aut^{-1}}{\BACHsmn \times \BACHsmn} \cdot \LOCINJHOM{\BACHsmn}{G}.
    \end{equation*}

    As these matrices and vectors are finite and both $\restr{\LeafAddingI}{\BACHsmn \times \BACHsmn}$ and $\restr{\Aut}{\BACHsmn \times \BACHsmn}^{-1}$ are still an invertible upper triangular matrix and an invertible diagonal matrix, respectively, we also get the forward direction:
    If a connected Berge-acyclic hypergraph $B$ has more than $n$ vertices in an edge, we trivially have $\LocInjHom(B,G) = 0$.
    Otherwise, we obtain the entry $\LocInjHom(B,G)$ by considering the class $\BACHsmn$ for the number of edges $m$ of $B$.
\end{proof}

\subsection{Simple Hypergraphs}
\label{sec:simplehypergraphs}
For a restriction of \Cref{th:hy:countingBACHHoms} to simple hypergraphs, consider a homomorphism $(h_V, h_E)$ from a hypergraph $G$ to a simple hypergraph $H$.
If $e, e' \in E(G)$ are parallel edges of $G$, i.e., $\EdgeFunOf{G}(e) = \EdgeFunOf{G}(e')$, then we have $\EdgeFunOf{H}(h_E(e)) = h_V(\EdgeFunOf{G}(e)) = h_V(\EdgeFunOf{G}(e')) = \EdgeFunOf{H}(h_E(e'))$, which implies $h_E(e) = h_E(e')$ since $H$ does not have parallel edges.
That is, parallel edges of $G$ have to be mapped to the same edge of $H$ since a homomorphism's mapping on edges is determined by its mapping on vertices up to parallel edges.
Hence, if we consider the simple hypergraph $G'$ obtained by merging parallel edges of $G$, then there is a one-to-one correspondence between the homomorphisms from $G$ to $H$ and these from $G'$ to $H$, and in particular, we have $\Hom(G, H) = \Hom(G', H)$.
Thus, for a simple hypergraph, it suffices to count homomorphisms from simple hypergraphs, and we obtain \Cref{th:hy:countingSimpleBACHHoms}, where $\SBACHs$ denotes the class of all connected Berge-acyclic simple hypergraphs.

\begin{corollary}
    \label{th:hy:countingSimpleBACHHoms}
    For all simple hypergraphs $G$ and $H$, the following are equivalent:
    \begin{enumerate}
        \item $\HOM{\SBACHs}{G} = \HOM{\SBACHs}{H}$.
        \item Color refinement does not distinguish $G$ and $H$.
    \end{enumerate}
\end{corollary}

For incidence homomorphisms, however, the situation is not as clear as these may map parallel edges to non-parallel ones.
However, with an interpolation argument, it is possible to prove that such a restriction can be made.

\begin{lemma}
    \label{le:hy:bachsmbachs}
    For all simple hypergraphs $G$ and $H$, the following are equivalent:
    \begin{enumerate}
        \item $\IHOM{\SBACHs}{G} = \IHOM{\SBACHs}{H}$.
        \item Color refinement does not distinguish $G$ and $H$.
    \end{enumerate}
\end{lemma}

\section{Directed Graphs}
\label{sec:di}

To prove that counting homomorphisms from DAGs suffices to characterize arbitrary directed graphs up to isomorphism, one could proceed in a similar fashion to \cite{Dell2018}, i.e., by defining an unfolding of a directed graph into a DAG and then proving the equivalence of counting homomorphisms and unfolding numbers.
This way, one obtains a characterization that is more intuitive than that of homomorphism counts, and one could show that an isomorphism between the directed graphs can be extracted from an isomorphism between appropriate unfoldings.
However, as the class of DAGs can also be defined in terms of homomorphism numbers, a proof using the algebraic properties of homomorphism counts turns out to be much simpler.

Lovász's second homomorphism-related work \cite{Lovasz1971} concerns the cancellation law among finite relational structures.
For the case of graphs, this asks whether a graph $K$ cancels out from the tensor products $G \TensorProduct K \iso H \TensorProduct K$, i.e., whether it satisfies the implication $G \TensorProduct K \iso H \TensorProduct K \implies G \iso H$ for all graphs $G$ and $H$.
Lovász gives the answer that this implication holds if and only if $K$ is not bipartite.
Moreover, from his work on the general case of finite relational structures, it follows that the transitive tournament $\TournamentOf{n}$ on $n$ vertices satisfies the cancellation law for directed graphs as long as $n \ge 3$.

To see how the cancellation law is related to homomorphism counts, observe that the class of DAGs can be defined as the class of all directed graphs that have a homomorphism into a transitive tournament.
Formally, if we let $\DAGs$ denote the class of DAGs and define $\DAGs_n \coloneqq \Set{G \mid \Hom(G, \TournamentOf{n}) > 0}$ for every $n \in \N$, then we have $\DAGs = \cup_{n \in \N} \DAGs_n$.
Then, using the facts that two directed graphs $G$ and $H$ are isomorphic if and only if we have $\Hom(F, G) = \Hom(F, H)$ for every directed graph $F$ \cite{Lovasz1967} and that $\Hom(F, G \TensorProduct H) = \Hom(F, G) \cdot \Hom(F, H)$ holds for all directed graphs $F$, $G$, and $H$ \cite{Lovasz1971}, we get that
\begin{align*}
    &G \TensorProduct \TournamentOf{n} \iso H \TensorProduct \TournamentOf{n}\\
    \iff &\forall F \ldotp \, \Hom(F, G \TensorProduct \TournamentOf{n}) = \Hom(F, H \TensorProduct \TournamentOf{n})\\
    \iff &\forall F \ldotp \, \Hom(F, G) \cdot \Hom(F, \TournamentOf{n}) = \Hom(F, H) \cdot \Hom(F, \TournamentOf{n})\\
    \iff &\HOM{\DAGs_n}{G} = \HOM{\DAGs_n}{H}
\end{align*}
holds for all directed graphs $G$ and $H$ and every $n \in \N$.
That is, tensor products with $\TournamentOf{n}$ are directly related to counting homomorphisms from $\DAGs_n$.

With the work of Lovász \cite{Lovasz1971}, this yields that two directed graphs $G$ and $H$ are isomorphic if and only if, for every DAG $D$, we have $\Hom(D, G) = \Hom(D, H)$.
More precisely, we obtain the even stronger statement that it suffices to count homomorphisms from the DAGs in $\DAGs_3$, i.e., from DAGs where the longest directed walk has length two.
For the case of undirected graphs, an analogous argument with the complete graph on three vertices $\KnOf{3}$, which is not bipartite, yields that arbitrary graphs can be characterized up to isomorphism by counting homomorphisms from all three-colorable graphs.

\section{Conclusion}

We have proven a generalization of the Tree Theorem for hypergraphs.
To this end, we have introduced a generalization of the color refinement algorithm for hypergraphs, which has lead to the notion of an incidence homomorphism.
By showing how incidence homomorphisms are related to homomorphisms, we have \enquote{reduced} the case of hypergraphs to the case of vertex-colored graphs.
For the case of directed graphs, we have revisited a result of Lovász, which shows that the class of DAGs is already too expressive to obtain an analogue of the Tree Theorem.

The central open question posed by our generalization of the Tree Theorem is whether it can further be generalized;
the Tree Theorem can be generalized to the $k$-dimensional Weisfeiler-Leman algorithm and graphs of treewidth at most $k$~\cite{Dell2018}.
An obvious attempt would be to consider the $k$-dimensional Weisfeiler-Leman algorithm on the colored incidence graphs of hypergraphs, in which case, however, our reduction does not generalize as we cannot restrict the identities of \Cref{le:hy:IHomLocInjIHom} and \Cref{le:hy:HomEInjHom} to hypergraphs whose incidence graphs have treewidth at most $k$;
merging vertices of a graph may increase its treewidth even when the merged vertices are part of the same neighborhood.
A way of interpreting this is that the treewidth of the incidence graph of a hypergraph $G$ is not a meaningful notion for $G$ since it mixes up the vertices and edges of $G$.

\bibliographystyle{plain}
\bibliography{bibliography}

\clearpage

\appendix

\section{Missing Proofs}

\begin{proof}[\Cref{le:hy:IHomLocInjIHom}]
    Let $G$ be a connected hypergraph and $H$ be a hypergraph.
    We have to prove that
    \begin{align}
        \IHom(G,H) = \sum_{G'} \EMergHom(G,G') \cdot \Aut^{-1}(G') \cdot \LocInjIHom(G', H) \tag{$\ast$}\label{le:hy:IHomLocInjIHom:equation}
    \end{align}
    holds, where the sum ranges over all connected hypergraphs $G'$.
    Note that a hypergraph $G'$ with $\EMergHom(G, G') > 0$ has at most $\abs{V(G)}$ vertices and exactly $\abs{E(G)}$ edges.
    There are only finitely many such hypergraphs $G'$, which means that the sum is finite and, thus, well-defined.
    Now, we consider pairs $\MerLocInjInc$ where
    \begin{itemize}
        \item $\Mer$ is a locally merging hom.\ from $G$ to a connected hypergraph $G'$, and
        \item $\LocInjInc$ is a locally injective incidence homomorphism from $G'$ to $H$.
    \end{itemize}
    For such a pair, we call $G'$ its type and say that it \textit{defines} the incidence homomorphism $(\LocInjIncV \circ \MerV, \LocInjIncE \circ \MerE)$ from $G$ to $H$.
    Observe that $(\LocInjIncV \circ \MerV, \LocInjIncE \circ \MerE)$ really is an incidence homomorphism as the composition of a homomorphism and an incidence homomorphism.
    To prove (\ref{le:hy:IHomLocInjIHom:equation}), we devise a mapping of incidence homomorphisms $\Inc$ from $G$ to $H$ to families $\MerLocInjIncFam$ of pairwise distinct pairs of type $G'$ that all define $\Inc$.
    We prove that every pair defining $\Inc$ is already part of this family.
    By doing so, we partition the pairs of type $G'$ into families of size $\Aut(G')$, and from those, we have a one-to-one correspondence to incidence homomorphisms from $G$ to $H$, which proves (\ref{le:hy:IHomLocInjIHom:equation}).

    To define the mapping $\Inc \mapsto \MerLocInjIncFam$, we first construct a connected hypergraph $\Merged$ and define $G'$ to be its isomorphism type.
    Intuitively, we obtain $\Merged$ by merging vertices that are part of the same edge and mapped to the same vertex by $\Inc$.
    To this end, we define $\Merged \coloneqq \QuotientGraph{G}{\MergeEquivOf{\IncV}}$ to be the quotient hypergraph of $G$ w.r.t.\ the equivalence relation $\MergeEquivOf{\IncV}$, i.e., we have
    \begin{align*}
        &V(\Merged) = \MergeEquivClassesOf{V(G)}{\IncV},& &E(\Merged) = E(G), \text{ and}& &\EdgeFunOf{\Merged}(e) = \Set{\EqClassOf{v} \mid v \in \EdgeFunOf{G}(e)}&
    \end{align*}
    for every $e \in E(\Merged)$, where $\EqClassOf{v}$ denotes the equivalence class of $v$ w.r.t.\ $\MergeEquivOf{\IncV}$ for every $v \in V(G)$.
    We let $G'$ be the isomorphism type of $\Merged$ and fix an arbitrary isomorphism $\Iso$ from $\Merged$ to $G'$.

    We define
    \begin{align*}
        &\MerV[m](v) \coloneqq \EqClassOf{v}& &\text{and}& &\MerE[m](e) \coloneqq e
    \end{align*}
    for every $v \in V(G)$ and every $e \in E(G)$, respectively, and claim that $\Mer[m]$ is a locally merging homomorphism from $G$ to $G'$:
    It is a homomorphism since we have
    \begin{align*}
        \MerV[m](\EdgeFunOf{G}(e)) = \Set{\EqClassOf{v} \mid v \in \EdgeFunOf{G}(e)}
        = \EdgeFunOf{\Merged}(e)
        = \EdgeFunOf{\Merged}(\MerE[m](e))
    \end{align*}
    for every $e \in E(G)$.
    Obviously, $\MerV[m]$ and $\MerE[m]$ are surjective and bijective, respectively.
    Moreover, the relations $\MergeEquivOf{\MerV[m]}$ and $\MergeEquivOf{\IncV}$ are the same, i.e., we have $u \MergeEquivOf{\MerV[m]} v \iff \MerV[m](u) = \MerV[m](v)$ for all $u,v \in V(G)$:
    If we have $u \MergeEquivOf{\MerV[m]} v$, then we directly get $\MerV[m](u) = \MerV[m](v)$ by the definition of $\MergeEquivOf{\MerV[m]}$.
    For the other direction,  we assume that we have $\MerV[m](u) = \MerV[m](v)$, i.e., $u \MergeEquivOf{\IncV} v$.
    By definition of $\MergeEquivOf{\IncV}$, there is a walk $v_0, e_1, \dots, v_k$ from $u$ to $v$ in $G$ with $\IncV(v_{i-1}) = \IncV(v_{i})$ for every $i \in \NumbersTo{k}$.
    For every $i \in \NumbersTo{k}$, this also means that $v_{i-1} \MergeEquivOf{\IncV} v_{i}$, which implies $\MerV[m](v_{i-1}) = \MerV[m](v_{i})$.
    Thus, this walk gives us that $u \MergeEquivOf{\MerV[m]} v$.
    
    When setting
    \begin{align*}
        &\LocInjIncV[m](\EqClassOf{v}) \coloneqq \IncV(v)& &\text{and}& &\LocInjIncE[m](e) \coloneqq \IncE(e)
    \end{align*}
    for every $v \in V(G)$ and every $e \in E(\Merged)$, respectively, $\LocInjIncV[m]$ is clearly well-defined, and $\LocInjInc[m]$ is a locally injective incidence homomorphism from $\Merged$ to $H$:
    We have
    \begin{align*}
        \LocInjIncV[m](\EdgeFunOf{\Merged}(e)) = \LocInjIncV[m](\Set{\EqClassOf{v} \mid v \in \EdgeFunOf{G}(e)}) = \IncV(\EdgeFunOf{G}(e)) \subseteq \EdgeFunOf{H}(\IncE(e)) = \EdgeFunOf{H}(\LocInjIncE[m](e))
    \end{align*}
    for every $e \in E(\Merged)$, which means that it is an incidence homomorphism.
    Note that, if $\Inc$ is a homomorphism, then we even have equality, which means that $\LocInjInc[m]$ is a homomorphism in this case;
    this is of interest when proving the analogous result for the special case of homomorphisms.
    To continue, we observe that $\LocInjInc$ is locally injective:
    Let $e \in E(\Merged)$ and $u,v \in V(G)$ with $\EqClassOf{u}, \EqClassOf{v} \in \EdgeFunOf{\Merged}(e)$.
    Then, there are $u' \in \EqClassOf{u}$ and $v' \in \EqClassOf{v}$ with $u',v' \in \EdgeFunOf{G}(e)$.
    If $\LocInjIncV[m](\EqClassOf{u}) = \LocInjIncV[m](\EqClassOf{v})$, then we have $\IncV(u) = \IncV(v)$, from which we get $\IncV(u') = \IncV(u) = \IncV(v) = \IncV(v')$.
    Since $u',v' \in \EdgeFunOf{G}(e)$, the definition of $\MergeEquivOf{\IncV}$ yields $u' \MergeEquivOf{\IncV} v'$, which gives us $u \MergeEquivOf{\IncV} v$, i.e.,  $\EqClassOf{u} = \EqClassOf{v}$, since $\MergeEquivOf{\IncV}$ is an equivalence relation.
    
    The locally merging homomorphism $\Mer[m]$ and the locally injective incidence homomorphism $\LocInjInc[m]$ define $\Inc$, i.e., we have
    \begin{align*}
        &\LocInjIncV[m](\MerV[m](v)) = \LocInjIncV[m](\EqClassOf{v}) = \IncV(v)& &\text{and}& &\LocInjIncE[m](\MerE[m](e)) = \LocInjIncE[m](e) = \IncE(e)
    \end{align*}
    for every $v \in V(G)$ and every $e \in E(G)$, respectively.
    We finally define
    \begin{align*}
        &\MerV[\pi] \coloneqq \pi_V \circ \IsoV \circ \MerV[m],&
        &\MerE[\pi] \coloneqq \pi_E \circ \IsoE \circ \MerE[m],\\
        &\LocInjIncV[\pi] \coloneqq \LocInjIncV[m] \circ \IsoV^{-1} \circ \pi_V^{-1} \text{, and}&
        &\LocInjIncE[\pi] \coloneqq \LocInjIncE[m] \circ \IsoE^{-1} \circ \pi_E^{-1}
    \end{align*}
    for an automorphism $(\pi_V, \pi_E) \in \AutGroup(G')$ and observe that $\MerLocInjInc[\pi]$ also defines $\Inc$.

    To prove that two different automorphisms $(\pi_V, \pi_E), (\pi'_V, \pi'_E) \in \AutGroup(G')$ result in different pairs, we distinguish two cases:
    If we have $\pi_V \neq \pi'_V$, then we also have $\MerV[\pi] \neq \MerV[\pi']$ because $\MerV[m]$ is surjective.
    Otherwise, if $\pi_E \neq \pi'_E$, then we get $\MerE[\pi] \neq \MerE[\pi']$ because $\MerE[m]$ is bijective.
    In either case, we have $\Mer[\pi] \neq \Mer[\pi']$ and, in particular, we get different pairs for $\pi$ and $\pi'$.

    Let $\MerLocInjIncPrime$ be a pair of type $G''$ defining $\Inc$, i.e., we have $\LocInjIncVPrime \circ \MerVPrime = \IncV$ and $\LocInjIncEPrime \circ \MerEPrime = \IncE$.
    By constructing an isomorphism $(\sigma_V, \sigma_E)$ from $\Merged$ to $G''$, we prove that $G' = G''$.
    From this isomorphism, we also obtain an automorphism $\pi \in \AutGroup(G')$ such that $\MerLocInjInc[\pi] = \MerLocInjIncPrime$.
    We define
    \begin{align*}
        &\sigma_V(\EqClassOf{v}) \coloneqq \MerVPrime(v)& &\text{and}& &\sigma_E(e) \coloneqq \MerEPrime(e)
    \end{align*}
    for every $v \in V(G)$ and every $e \in E(\Merged)$, respectively.
    Trivially, $\sigma_E$ is bijective, and to prove that $\sigma_V$ is well-defined and bijective, it suffices that we observe that the relations $\MergeEquivOf{\IncV}$ and $\MergeEquivOf{\MerVPrime}$ are the same as this yields
    \begin{align*}
        u \MergeEquivOf{\IncV} v \iff u \MergeEquivOf{\MerVPrime} v \iff \MerVPrime(u) = \MerVPrime(v)
    \end{align*}
    for all $u,v \in V(G)$ since $\MerPrime$ is a locally merging homomorphism.
    Intuitively, the relations are the same because, if two vertices within an edge $e \in E(G)$ collapse, then this already has to happen in $\MerPrime$ since $\LocInjIncPrime$ is locally injective.
    Formally, we have
    \begin{align*}
        \IncV(u) = \IncV(v) \iff \LocInjIncVPrime(\MerVPrime(u)) = \LocInjIncVPrime(\MerVPrime(v)) \iff \MerVPrime(u) = \MerVPrime(v)
    \end{align*}
    for all $u,v \in \EdgeFunOf{G}(e)$ and every $e \in E(G)$ since $\LocInjIncPrime$ is locally injective.
    Then, by definition of the relations, we also get that $\MergeEquivOf{\IncV}$ and $\MergeEquivOf{\MerVPrime}$ are the same.
    To see that $(\sigma_V, \sigma_E)$ really is an isomorphism, we observe that we have
    \begin{align*}
        \sigma_V(\EdgeFunOf{\Merged}(e)) = \sigma_V(\Set{\EqClassOf{v} \mid v \in \EdgeFunOf{G}(e)}) = \MerVPrime(\EdgeFunOf{G}(e)) = \EdgeFunOf{G''}(\MerEPrime(e)) = \EdgeFunOf{G''}(\sigma_E(e))
    \end{align*}
    for every edge $e \in E(\Merged)$ since $\MerPrime$ is, in particular, a homomorphism, which yields $G' = G''$.

    We consider the automorphism $(\pi_V, \pi_E)$ of $G'$ obtained by setting
    \begin{align*}
        &\pi_V \coloneqq \sigma_V \circ \IsoV^{-1}& &\text{and} & &\pi_E \coloneqq \sigma_E \circ \IsoE^{-1},
    \end{align*}
    and we claim that $\MerLocInjInc[\pi] = \MerLocInjIncPrime$.
    We have
    \begin{align*}
        \MerV[\pi](v) = \sigma_V(\IsoV^{-1}(\IsoV(\MerV[m](v)))) = \sigma_V(\MerV[m](v)) = \sigma_V(\EqClassOf{v}) = \MerVPrime(v)
    \end{align*}
    for every $v \in V(G)$ and
    \begin{align*}
        \MerE[\pi](e) = \sigma_E(\IsoE^{-1}(\IsoE(\MerE[m](e)))) = \sigma_E(\MerE[m](e)) = \sigma_E(e) = \MerEPrime(e)
    \end{align*}
    for every $e \in E(G)$, which proves that $\Mer[\pi] = \MerPrime$.
    Moreover, we have
    \begin{align*}
        \LocInjIncV[\pi](\sigma_V(\EqClassOf{v})) &= \LocInjIncV[m](\IsoV^{-1}(\IsoV(\sigma_V^{-1}(\sigma_V(\EqClassOf{v})))))\\
        &= \LocInjIncV[m](\EqClassOf{v})\\
        &= \IncV(v)\\
        &= \LocInjIncVPrime(\MerVPrime(v))\\
        &= \LocInjIncVPrime(\sigma_V(\EqClassOf{v}))
    \end{align*}
    for every $v \in V(G)$ and
    \begin{align*}
        \LocInjIncE[\pi](\sigma_E(e)) &= \LocInjIncE[m](\IsoE^{-1}(\IsoE(\sigma_E^{-1}(\sigma_E(e)))))\\
        &= \LocInjIncE[m](e)\\
        &= \IncE(e)\\
        &= \LocInjIncEPrime(\MerEPrime(e))\\
        &= \LocInjIncEPrime(\sigma_E(e))
    \end{align*}
    for every $e \in V(\Merged)$, which proves that $\LocInjInc[\pi] = \LocInjIncPrime$ since both $\sigma_V$ and $\sigma_E$ are surjective.

    To prove that $\EMergHom$ is a lower triangular matrix, let $G$ and $H$ be distinct connected hypergraphs with $\abs{V(G)} + \abs{E(G)} \le \abs{V(H)} + \abs{E(H)}$.
    Assume that $\EMergHom(G, H) > 0$, i.e., that there is a locally merging homomorphism $(h_V, h_E)$ from $G$ to $H$.
    Since $h_V$ is surjective, we get that $\abs{V(G)} \ge \abs{V(H)}$, and since $h_E$ is bijective, this yields $\abs{V(G)} = \abs{V(H)}$ with the assumption.
    Hence, $h_V$ is bijective, which implies that $(h_V, h_E)$ is an isomorphism and contradicts the assumption that $G$ and $H$ are distinct.
    Moreover, the diagonal entries of $\EMergHom$ are non-zero since we have $\EMergHom(G, G) = \Aut(G) > 0$ for every connected hypergraph $G$, i.e., $\EMergHom$ is invertible.
\end{proof}

\begin{proof}[\Cref{le:hy:HomEInjHom}]
    This is a special case of the proof of \Cref{le:hy:IHomLocInjIHom}:
    Instead of considering pairs $\MerLocInjInc$ where $\Mer$ is a locally merging homomorphism and $\LocInjInc$ a locally injective incidence homomorphism, we additionally require $\LocInjInc$ to be a homomorphism.
    Then, such a pair defines a homomorphism as the composition of two homomorphisms, and the restriction of the mapping $\Inc \mapsto \MerLocInjIncFam$ to homomorphisms $\Inc$ only yields such pairs.
    Restricting the mapping to a subset of incidence homomorphisms does not change the facts that all pairs in the family $\MerLocInjIncFam$ are distinct and that every pair defining $\Inc$ is part of this family.
\end{proof}

\begin{proof}[\Cref{le:hy:LocInjIHomLocInjHom}]
    Let $G$ be a connected hypergraph, and let $H$ be a hypergraph.
    We have to prove that
    \begin{align}
        \LocInjIHom(G,H) = \sum_{G'} \LeafAddingI(G,G') \cdot \Aut^{-1}(G') \cdot \LocInjHom(G', H) \tag{$\ast$} \label{le:hy:LocInjIHomLocInjHom:equation}
    \end{align}
    holds, where the sum ranges over all connected hypergraphs $G'$.
    Note that a hypergraph $G'$ with $\LeafAddingI(G,G') > 0$ has exactly as many edges as $G$.
    Moreover, if we additionally have $\LocInjHom(G', H) > 0$, then $G'$ has at most $\abs{V(H)}$ vertices per edge.
    Hence, we also get that the number of its vertices is bounded by $\abs{E(G)} \cdot \abs{V(H)}$ as long as it is not an isolated vertex, which means that there are only finitely many such graphs $G'$, and we get that the sum is finite, and thus, well-defined.
    We consider pairs $\LAddLocInj$ where
    \begin{itemize}
        \item $\LAdd$ is a leaf-adding in.\ hom.\ from $G$ to a connected hypergraph $G'$, and
        \item $\LocInj$ is a locally injective homomorphism from $G'$ to $H$.
    \end{itemize}
    For such a pair, we call $G'$ its type and say that it \textit{defines} the locally injective incidence homomorphism $(\LocInjV \circ \LAddV, \LocInjE \circ \LAddE)$ from $G$ to $H$.
    Observe that $(\LocInjV \circ \LAddV, \LocInjE \circ \LAddE)$ really is an incidence homomorphism as the composition of two incidence homomorphisms and that it is locally injective as $\LAddV$ is injective and $\LocInj$ is locally injective.
    To prove (\ref{le:hy:LocInjIHomLocInjHom:equation}), we devise a mapping of locally injective incidence homomorphisms $\LocInjInc$ from $G$ to $H$ to families $\LAddLocInjFam$ of pairwise distinct pairs of type $G'$ that all define $\LocInjInc$.
    We prove that every pair defining $\LocInjInc$ is already part of this family.
    By doing so, we partition the pairs of type $G'$ into families of size $\Aut(G')$, and from those, we have a one-to-one correspondence to locally injective incidence homomorphisms from $G$ to $H$, which proves (\ref{le:hy:LocInjIHomLocInjHom:equation}).

    To define the mapping $\LocInjInc \mapsto \LAddLocInjFam$, we first construct a connected hypergraph $\Balanced$ and define $G'$ to be its isomorphism type.
    Intuitively, we obtain $\Balanced$ from $G$ by filling up its edges with fresh leaves until $\LocInjInc$ becomes a homomorphism.
    \newcommand{\NumDummies}[1]{n_{#1}}%
    Formally, for every $e \in E(G)$, we let
    \begin{align*}
        \NumDummies{e} \coloneqq \abs{\EdgeFunOf{H}(\LocInjIncE(e)) \setminus \LocInjIncV(\EdgeFunOf{G}(e))}
    \end{align*}
    be the number of leaves we add to the edge $e$, and for every $i \in \NumbersTo{\NumDummies{e}}$, we let $v^{e}_{i}$ be a fresh vertex.
    We define $\Balanced$ by setting
    \begin{align*}
        &V(\Balanced) \coloneqq V(G) \cup \Set{v^{e}_{i} \mid e \in E(G),\, i \in \NumbersTo{\NumDummies{e}}},& &E(\Balanced) \coloneqq E(G),
    \end{align*}
    and
    \begin{align*}
        \EdgeFunOf{\Balanced}(e) \coloneqq \EdgeFunOf{G}(e) \cup \Set{v^{e}_{i} \mid i \in \NumbersTo{\NumDummies{e}}}
    \end{align*}
    for every $e \in E(\Balanced)$.
    We let $G'$ be the isomorphism type of $\Balanced$ and fix an arbitrary isomorphism $(\IsoBalV, \IsoBalE)$ from $\Balanced$ to $G'$.

    By setting
    \begin{align*}
        &\LAddV[a](v) \coloneqq v& &\text{and}& &\LAddE[a](e) \coloneqq e
    \end{align*}
    for every $v \in V(G)$ and every $e \in E(G)$, respectively, we get that $\LAdd[a]$ is a leaf-adding incidence homomorphism from $G$ to $\Balanced$:
    we have
    \begin{align*}
        \LAddV[a](\EdgeFunOf{G}(e)) = \EdgeFunOf{G}(e) \subseteq \EdgeFunOf{\Balanced}(e) = \EdgeFunOf{\Balanced}(\LAddE[a](e))
    \end{align*}
    and
    \begin{align*}
        \LAddV[a]^{-1}(\EdgeFunOf{\Balanced}(\LAddE[a](e)) \cap \ImageOf{\LAddV[a]}) = \LAddV[a]^{-1}(\EdgeFunOf{\Balanced}(e) \cap \ImageOf{\LAddV[a]}) = \LAddV[a]^{-1}(\EdgeFunOf{G}(e)) = \EdgeFunOf{G}(e)
    \end{align*}
    for every $e \in E(\Balanced)$, i.e., $\LAdd[a]$ is a strong incidence homomorphism, and the other properties of a leaf-adding incidence homomorphism are trivially satisfied.
    By setting
    \begin{align*}
        &\LocInjV[a](v) \coloneqq \LocInjIncV(v)& &\text{and}& &\LocInjE[a](e) \coloneqq \LocInjIncE(e)
    \end{align*}
    for every $v \in V(G)$ and every $e \in E(\Balanced)$, respectively, and additionally letting $\LocInjV[a]$ map the vertices in $\Set{v^{e}_{i} \mid i \in \NumbersTo{\NumDummies{e}}}$ bijectively to these in $\EdgeFunOf{H}(\LocInjIncE(e)) \setminus \LocInjIncV(\EdgeFunOf{G}(e))$ for every $e \in E(\Balanced)$, we obtain a locally injective homomorphism from $\Balanced$ to $H$:
    We have
    \begin{align*}
        \LocInjV[a](\EdgeFunOf{\Balanced}(e)) &= \LocInjV[a](\EdgeFunOf{G}(e) \cup \Set{v^{e}_{i} \mid i \in \NumbersTo{\NumDummies{e}}})\\
        &= \LocInjIncV(\EdgeFunOf{G}(e)) \cup (\EdgeFunOf{H}(\LocInjIncE(e)) \setminus \LocInjIncV(\EdgeFunOf{G}(e)))\\
        &= \EdgeFunOf{H}(\LocInjIncE(e))
    \end{align*}
    for every $e \in E(\Balanced)$, i.e., $\LocInj[a]$ is a homomorphism.
    Moreover, since $\LocInj$ is locally injective, the definition of $\LocInjInc[a]$ immediately yields that it also is locally injective as the leaves added to an edge $e$ are bijectively mapped to vertices not hit before by the vertices of $e$.
    The leaf-adding incidence homomorphism $\LAdd[a]$ and the locally injective homomorphism $\LocInj[a]$ define $\LocInjInc$, i.e.,
    we have
    \begin{align*}
        &\LocInjV[a](\LAddV[a](v)) = \LocInjV[a](v) = \LocInjIncV(v)& &\text{and}& &\LocInjE[a](\LAddE[a](e)) = \LocInjE[a](e) = \LocInjIncE(e)
    \end{align*}
    for every $v \in V(G)$ and every $e \in E(G)$, respectively.
    We finally define
    \begin{align*}
        &\LAddV[\pi] \coloneqq \pi_V \circ \IsoBalV \circ \LAddV[a],&
        &\LAddE[\pi] \coloneqq \pi_E \circ \IsoBalE \circ \LAddE[a],\\
        &\LocInjV[\pi] \coloneqq \LocInjV[a] \circ \IsoBalV^{-1} \circ \pi_V^{-1} \text{, and}&
        &\LocInjE[\pi] \coloneqq \LocInjE[a] \circ \IsoBalE^{-1} \circ \pi_E^{-1}
    \end{align*}
    for an automorphism $(\pi_V, \pi_E) \in \AutGroup(G')$ and observe that $\LAddLocInj[\pi]$ also defines $\LocInjInc$.

    To prove that two different automorphisms $(\pi_V, \pi_E), (\pi'_V, \pi'_E) \in \AutGroup(G')$ result in different pairs, we distinguish three cases:
    If $\pi_E \neq \pi'_E$, then we have $\LAddE[\pi] \neq \LAddE[\pi']$ since $\LAddE[a]$ and $\IsoBalE$ are surjective.
    Otherwise, in the remaining two cases, we have $\pi_E = \pi'_E$ and $\pi_V \neq \pi'_V$.
    For the second case, assume that there is a vertex $v \in \IsoBalV(\LAddV[a](V(G)))$ with $\pi_V(v) \neq \pi'_V(v)$.
    Then, we directly have $\LAddV[\pi] \neq \LAddV[\pi']$.
    For the third case, assume that $\pi_V$ and $\pi'_V$ are the same on $\IsoBalV(\LAddV[a](V(G)))$.
    We observe that $\pi_V^{-1}$ and $\pi'^{-1}_V$ are the same on all non-leaves of $G'$:
    As $\LAdd[a]$ is leaf-adding and $\IsoBal$ an isomorphism, the assumption of the case implies that $\pi_V$ and $\pi'_V$ differ only on leaves and are the same on all non-leaves.
    Since non-leaves are mapped to non-leaves by $\pi_V$ and $\pi'_V$, their inverses also are the same on non-leaves.
    To continue, because $\pi_V \neq \pi'_V$, there is vertex $v \in V(G')$ with $\pi^{-1}_V(v) \neq \pi'^{-1}_V(v)$, which has to be a leaf and, hence, is connected to an edge $e \in E(G')$.
    By $\pi_E = \pi'_E$, we also have $\pi^{-1}_E = \pi'^{-1}_E$ and get that both $\pi^{-1}_V(v)$ and $\pi'^{-1}_V(v)$ are elements of the edge $\pi^{-1}_E(e) = \pi'^{-1}_E(e)$, which implies that
    \begin{align*}
       \LocInjV[\pi](v) = \LocInjV[a](\IsoBalV^{-1}(\pi_V^{-1}(v))) \neq \LocInjV[a](\IsoBalV^{-1}(\pi'^{-1}_V(v))) = \LocInjV[\pi'](v)
    \end{align*}
    because $\LocInjV[a]$ is locally injective and $\IsoBal$ an isomorphism.
    Therefore, we have $\LocInjV[\pi] \neq \LocInjV[\pi']$.

    Let $\LAddLocInjPrime$ be a pair of type $G''$ defining $\LocInjInc$, i.e., we have $\LocInjVPrime \circ \LAddVPrime = \LocInjIncV$ and $\LocInjEPrime \circ \LAddEPrime = \LocInjIncE$.
    By constructing an isomorphism $(\sigma_V, \sigma_E)$ from $\Balanced$ to $G''$, we prove that $G' = G''$.
    From this isomorphism, we also obtain an automorphism $\pi \in \AutGroup(G')$ such that $\LAddLocInj[\pi] = \LAddLocInjPrime$.
    We let
    \begin{align*}
        \sigma_E(e) \coloneqq \LAddEPrime(\LAddE[a]^{-1}(e)) = \LAddEPrime(e)
    \end{align*}
    for every $e \in E(\Balanced)$, which is a bijection by the definition of a leaf-adding incidence homomorphism.
    In contrast, defining $\sigma_V$ is not straightforward due to the vertices of $G''$ not hit by $\LAddVPrime$, i.e., the added leaves.
    For $v \in V(G)$, we set
    \begin{align*}
        \sigma_V(v) \coloneqq \LAddVPrime(v)
    \end{align*}
    and obtain a bijection from $V(G)$ to $\LAddVPrime(V(G))$, which we extend to $V(\Balanced)$ in the following.
    Let $e \in E(G)$ and consider the corresponding edge $\LAddEPrime(e)$ in $G''$.
    Since $\LocInjPrime$ is a locally injective homomorphism, $\LocInjVPrime$ bijectively maps $\EdgeFunOf{G''}(\LAddEPrime(e))$ to $\EdgeFunOf{H}(\LocInjEPrime(\LAddEPrime(e)))$.
    Moreover, the subset $\LAddVPrime(\EdgeFunOf{G}(e))$ is mapped bijectively to $\LocInjVPrime(\LAddVPrime(\EdgeFunOf{G}(e)))$.
    Therefore, the restriction of $\LocInjVPrime$ to $\EdgeFunOf{G''}(\LAddEPrime(e)) \setminus \LAddVPrime(\EdgeFunOf{G}(e))$ is a bijection to $\EdgeFunOf{H}(\LocInjEPrime(\LAddEPrime(e))) \setminus \LocInjVPrime(\LAddVPrime(\EdgeFunOf{G}(e)))$, i.e., to $\EdgeFunOf{H}(\LocInjIncE(e)) \setminus \LocInjIncV(\EdgeFunOf{G}(e))$.
    Recall that, by definition, $\LocInjE[a]$ maps $\Set{v^{e}_{i} \mid i \in \NumbersTo{\NumDummies{e}}}$ bijectively to $\EdgeFunOf{H}(\LocInjIncE(e)) \setminus \LocInjIncV(\EdgeFunOf{G}(e))$.
    Thus, by setting
    \begin{align*}
        \sigma_V(v^e_i) \coloneqq \restr{\LocInjVPrime}{\EdgeFunOf{G''}(\LAddEPrime(e)) \setminus \LAddVPrime(\EdgeFunOf{G}(e))}^{-1}(\LocInjV[a](v^e_i))
    \end{align*}
    for every $i \in \NumbersTo{\NumDummies{e}}$, we bijectively map $\Set{v^{e}_{i} \mid i \in \NumbersTo{\NumDummies{e}}}$ to $\EdgeFunOf{G''}(\LAddEPrime(e)) \setminus \LAddVPrime(\EdgeFunOf{G}(e))$.

    To prove that $\sigma_V$ is a bijection from $V(\Balanced)$ to $V(G'')$, it suffices to observe that
    \begin{align*}
        \dot{\bigcup_{e \in E(G)}} \EdgeFunOf{G''}(\LAddEPrime(e)) \setminus \LAddVPrime(\EdgeFunOf{G}(e)) = V(G'') \setminus \LAddVPrime(V(G)).
    \end{align*}
    For an edge $e \in E(G)$, Requirement~\ref{en:hy:leafadding:strong} of the definition of a leaf-adding incidence homomorphism implies $\EdgeFunOf{G''}(\LAddEPrime(e)) \cap \ImageOf{\LAddVPrime} \subseteq \LAddVPrime(\EdgeFunOf{G}(e))$ or, equivalently, $\EdgeFunOf{G''}(\LAddEPrime(e)) \setminus \LAddVPrime(\EdgeFunOf{G}(e)) \subseteq V(G'') \setminus \ImageOf{\LAddVPrime}$, which already yields the inclusion \enquote{$\subseteq$}.
    Furthermore, Requirement~\ref{en:hy:leafadding:notHitLeaf} of the definition yields that the vertices $\EdgeFunOf{G''}(\LAddEPrime(e)) \setminus \LAddVPrime(\EdgeFunOf{G}(e)) \subseteq V(G'')$ are leaves, which implies that the sets being united are disjoint as $\LAddEPrime$ is injective.
    The inclusion \enquote{$\supseteq$} is also directly apparent from the definition of a leaf-adding incidence homomorphism as every vertex not hit by $\LAddVPrime$ is a leaf and, thus, element of an edge.
    Finally, we observe that $\sigma$ is an isomorphism and, thus, $G' = G''$ as, for every $e \in E(\Balanced)$, we have
    \begin{align*}
        \sigma_V(\EdgeFunOf{\Balanced}(e)) &= \sigma_V(\EdgeFunOf{G}(e) \cup \Set{v^{e}_{i} \mid i \in \NumbersTo{\NumDummies{e}}})\\
        &= \LAddVPrime(\EdgeFunOf{G}(e)) \cup (\EdgeFunOf{G''}(\LAddEPrime(e)) \setminus \LAddVPrime(\EdgeFunOf{G}(e)))\\
        &= \EdgeFunOf{G''}(\LAddEPrime(e))\\
        &= \EdgeFunOf{G''}(\sigma_E(e)).
    \end{align*}

    We consider the automorphism $(\pi_V, \pi_E)$ of $G'$ obtained by setting
    \begin{align*}
        &\pi_V \coloneqq \sigma_V \circ \IsoBalV^{-1}& &\text{and} & &\pi_E \coloneqq \sigma_E \circ \IsoBalE^{-1},
    \end{align*}
    and we claim that $\LAddLocInj[\pi] = \LAddLocInjPrime$.
    We have
    \begin{align*}
        \LAddV[\pi](v) = \sigma_V(\IsoBalV^{-1}(\IsoBalV(\LAddV[a](v)))) = \sigma_V(\LAddV[a](v)) = \sigma_V(v) = \LAddVPrime(v)
    \end{align*}
    for every $v \in V(G)$, which proves that $\LAddV[\pi] = \LAddVPrime$, and
    \begin{align*}
        \LAddE[\pi](e) = \sigma_E(\IsoBalE^{-1}(\IsoBalE(\LAddE[a](e)))) = \sigma_E(\LAddE[a](e)) = \sigma_E(e) = \LAddEPrime(e)
    \end{align*}
    for every $e \in E(G)$, which proves that $\LAddE[\pi] = \LAddEPrime$.
    Furthermore, we have
    \begin{align*}
        \LocInjV[\pi](\sigma_V(v))
        &= \LocInjV[a](\IsoBalV^{-1}(\IsoBalV(\sigma_V^{-1}(\sigma_V(v))))\\
        &= \LocInjV[a](v)\\
        &= \LocInjIncV(v)\\
        &= \LocInjVPrime(\LAddVPrime(v))\\
        &= \LocInjVPrime(\sigma_V(v))
    \end{align*}
    for every $v \in V(\Balanced)$ and
    \begin{align*}
        \LocInjV[\pi](\sigma_V(v^e_i))
        &= \LocInjV[a](\IsoBalV^{-1}(\IsoBalV(\sigma_V^{-1}(\sigma_V(v^e_i))))\\
        &= \LocInjV[a](v^e_i)\\
        &= \LocInjVPrime(\restr{\LocInjVPrime}{\EdgeFunOf{G''}(\LAddEPrime(e)) \setminus \LAddVPrime(\EdgeFunOf{G}(e))}^{-1}(\LocInjV[a](v^e_i)))\\
        &= \LocInjVPrime(\sigma_V(v^e_i))
    \end{align*}
    for every $e \in E(G)$ and $i \in \NumbersTo{\NumDummies{e}}$, which proves that $\LocInjV[\pi] = \LocInjVPrime$.
    Finally, we have
    \begin{align*}
        \LocInjE[\pi](\sigma_E(e))
        &= \LocInjE[a](\IsoBalE^{-1}(\IsoBalE(\sigma_E^{-1}(\sigma_E(e))))\\
        &= \LocInjE[a](e)\\
        &= \LocInjIncE(e)\\
        &= \LocInjEPrime(\LAddEPrime(e))\\
        &= \LocInjEPrime(\sigma_E(e))
    \end{align*}
    for every $e \in E(\Balanced)$, which proves that $\LocInjE[\pi] = \LocInjEPrime$.

    To prove that $\LeafAddingI$ is an upper triangular matrix, let $G$ and $H$ be distinct connected hypergraphs with $\abs{V(G)} + \abs{E(G)} \ge \abs{V(H)} + \abs{E(H)}$.
    Assume that $\EMergHom(G, H) > 0$, i.e., that there is a leaf-adding incidence homomorphism $(h_V, h_E)$ from $G$ to $H$.
    We get that $\abs{V(G)} \le \abs{V(H)}$ because $h_V$ is injective, and since $h_E$ is bijective, this yields $\abs{V(G)} = \abs{V(H)}$ with the assumption.
    Hence, $h_V$ is bijective, which implies that $(h_V, h_E)$ is a homomorphism and, furthermore, an isomorphism and contradicts the assumption that $G$ and $H$ are distinct.
    Moreover, the diagonal entries of $\LeafAddingI$ are non-zero since we have $\LeafAddingI(G, G) = \Aut(G) > 0$ for every connected hypergraph $G$, i.e., $\LeafAddingI$ is invertible.
\end{proof}

\begin{proof}[\Cref{le:hy:bachsmbachs}]
    With \Cref{th:hy:countingMBACHIncHoms}, it suffices to prove that, for all simple hypergraphs $G$ and $H$, we have $\IHOM{\BACHs}{G} = \IHOM{\BACHs}{H}$ if and only if $\IHOM{\SBACHs}{G} = \IHOM{\SBACHs}{H}$, where the forward direction is trivial.

    To introduce some notation, let $G$ be a hypergraph with $n \coloneqq \abs{V(G)}$ vertices, and let $v \in V(G)$ be a vertex.
    The \textit{degree} of $v$ is given by
    \begin{align*}
        \DegreeOf{v} \coloneqq \abs{\Set{e \in E(G) \mid v \in \EdgeFunOf{G}(e)}}
    \end{align*}
    and, for every $i \in \NumbersTo{n}$, the \textit{$i$-degree} of $v$ is given by
    \begin{align*}
        \IDegreeOf{i}{v} \coloneqq \abs{\MultiSet{e \in E(G) \mid v \in \EdgeFunOf{G}(e),\, \abs{\EdgeFunOf{G}(e)} = i}}.
    \end{align*}
    We define the \textit{degree sequence} $\DegreeSeqOf{v} \coloneqq (\IDegreeOf{1}{v}, \dots, \IDegreeOf{n}{v}) \in \N^{n}$ of $v$ and note that we have $\sum_{i = 1}^{n} \DegreeSeqOf{v}_{i} = \DegreeOf{v}$.
    Observe that, if $G$ is simple, then we have $\DegreeOf{v} \in \Set{0, \dots, 2^{n-1}}$ and $\IDegreeOf{i}{v} \in \Set{0, \dots, \binom{n-1}{i-1}}$ for every $i \in \NumbersTo{n}$.
    In this case, setting $\DegreeSeqsOf{n} \coloneqq \NumbersTo{\binom{n-1}{0}}\ \times \dots \times \NumbersTo{\binom{n-1}{n-1}}$ yields a finite set of all degree sequences a vertex may have.
    For our interpolation argument, however, letting $\DegreeSeqsOf{n}$ be the larger set $\NumbersTo{2^{n-1}}^{n}$ would also suffice.

    To prove the non-trivial direction, we show that, for every hypergraph $G$, every entry of $\IHOM{\BACHs}{G}$ is determined by the entries of $\IHOM{\SBACHs}{G}$.
    To this end, let $B$ be a connected Berge-acyclic hypergraph, and let $G$ be a hypergraph.
    We prove the statement by induction on the number of vertices of $B$ that have parallel loops, which are the only form of parallel edges in $B$ since it is Berge-acyclic.
    If $B$ does not have any parallel loops, the claim trivially holds.

    For the inductive step, we fix a vertex $u \in V(B)$ with parallel loops, i.e., we have $\abs{\LeavesOf{u}} \ge 2$ for the set $\LeavesOf{u} \coloneqq \Set{e \in E(B) \mid \EdgeFunOf{B}(e) = \Set{u}}$ of loops at $u$.
    \newcommand{\FreshEdge}{e_\ell}%
    We define
    \begin{align*}
        &V(B') \coloneqq V(B),& &E(B') \coloneqq (E(B) \setminus \LeavesOf{u}) \cupdot \Set{\FreshEdge},
    \end{align*}
    where $\FreshEdge$ is a fresh edge, and
    \begin{align*}
        \EdgeFunOf{B'}(e) \coloneqq
        \begin{cases}
            \EdgeFunOf{B}(e) &\text{ if } e \in E(B) \setminus \LeavesOf{u},\\
            \Set{u} &\text{ if } e = \FreshEdge
        \end{cases}
    \end{align*}
    for every $e \in E(B')$, i.e., we obtain $B'$ from $B$ by merging the loops at $u$ into a single one.
    Let $n \coloneqq \abs{V(G)} = \IHom(\KnOf{1},G)$ be the number of vertices of $G$, and for every~$i \ge 0$, let $\IHomDegree{B'}{G}{i}$ denote the number of incidence homomorphisms from $B'$ to $G$ that map $u$ to a vertex of degree $i$.
    Note that, since $u$ has a loop in $B'$, we have $\IHomDegree{B'}{G}{0} = 0$.
    Moreover, a vertex of $G$ has at most degree $2^{n-1}$, which means that $\IHomDegree{B'}{G}{i} = 0$ for every $i > 2^{n-1}$.
    Observe that we have
    \begin{align*}
        \IHom(B, G) = \sum_{i = 1}^{2^{n-1}} \IHomDegree{B'}{G}{i} \cdot i ^{\abs{\LeavesOf{u}}-1},
    \end{align*}
    i.e., it suffices to prove that the values $\IHomDegree{B'}{G}{1}, \dots, \IHomDegree{B'}{G}{2^{n-1}}$ are determined by $\IHOM{\SBACHs}{G}$.
    \newcommand{\Bkl}{B_{r, s}}%

    \begin{figure}[htbp]
    	\centering
        \scalebox{0.8}{
    	\begin{tikzpicture}[node distance = \VertexDistance cm]
    		\node[hypervertex, label={270:$u$}] (u) {};

    		\node[] at ($(u) + (-0.5*\VertexDistance,\EdgeDistance)$) (oldL) {};
            \node[] at ($(u) + (0.5*\VertexDistance,\EdgeDistance)$) (oldR) {};

            \node[hypervertex, label={270:$v^1_1$}] at ($(u) + (1.5 * \VertexDistance, 0)$) (v11) {};
            \node[hypervertex, label={270:$v^1_s$}] at ($(u) + (2.5 * \VertexDistance, 0)$) (v1s) {};
            \node[hyperedge, label={90:$e_1$}] at ($(v11)!0.5!(v1s) + (0,\EdgeDistance)$) (e1) {};

            \node[] at ($(v11)!0.5!(v1s)$) {$\ldots\ldots$};

            \draw[hyperedgeedge] (u) to (e1);
            \draw[hyperedge] (e1) to (v11);
            \draw[hyperedge] (e1) to (v1s);

            \node[hypervertex, label={270:$v^r_1$}] at ($(u) + (4 * \VertexDistance, 0)$) (vr1) {};
            \node[hypervertex, label={270:$v^r_s$}] at ($(u) + (5 * \VertexDistance, 0)$) (vrs) {};
            \node[hyperedge, label={90:$e_r$}] at ($(vr1)!0.5!(vrs) + (0,\EdgeDistance)$) (er) {};

            \node[] at ($(vr1)!0.5!(vrs)$) {$\ldots\ldots$};

            \draw[hyperedgeedge] (u) to (er);
            \draw[hyperedge] (er) to (vr1);
            \draw[hyperedge] (er) to (vrs);

            \node[] at ($(e1)!0.5!(er)$) {$\ldots\ldots\ldots\ldots$};
            \draw[hyperedgeedge] (u) to  ($(u)!0.7!(oldL)$);
            \draw[hyperedgeedge, dotted] ($(u)!0.7!(oldL)$) to  ($(u)!0.95!(oldL)$);

            \draw[hyperedgeedge] (u) to  ($(u)!0.7!(oldR)$);
            \draw[hyperedgeedge, dotted] ($(u)!0.7!(oldR)$) to  ($(u)!0.95!(oldR)$);

            \node[] at ($(oldL)!0.5!(oldR) - (0, 0.15 * \EdgeDistance)$) {$\ldots\ldots$};
    	\end{tikzpicture}
        }
    	\caption{Construction of $\Bkl$ in the proof of \Cref{le:hy:bachsmbachs}}
    	\label{fig:hy:interpolationconstruction}
    \end{figure}
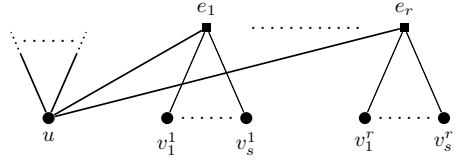

    To obtain these values via an interpolation argument, for every $r \ge 0$ and every $s > 0$, we define $\Bkl$ from $B'$ by adding $r$ edges that each contain the vertex $u$ and $s$ fresh vertices.
    Formally, we let
    \begin{align*}
        V(\Bkl) \coloneqq V(B') \cupdot \Set{v^{i}_{j} \mid i \in \NumbersTo{r}, \, j \in \NumbersTo{s}},
    \end{align*}
    where $v^{i}_{j}$ is a fresh vertex for all $i \in \NumbersTo{r}, \,j \in \NumbersTo{s}$,
    \begin{align*}
        E(\Bkl) \coloneqq E(B') \cupdot \Set{e_1, \dots, e_r},
    \end{align*}
    where $e_1, \dots, e_r$ are fresh edges, and
    \begin{align*}
        \EdgeFunOf{\Bkl}(e) \coloneqq
        \begin{cases}
            \EdgeFunOf{B'}(e) & \text{ if }e \in E(B'),\\
            \Set{u} \cup \Set{v^{i}_{j} \mid j \in \NumbersTo{s}} & \text{ if } e = e_i \text{ for } i \in \NumbersTo{r}
        \end{cases}
    \end{align*}
    for every $e \in E(\Bkl)$.
    Since we require $s$ to be strictly greater than zero, the vertex $u$ does not have any parallel loops, which means that $\Bkl$ has fewer vertices with parallel loops than $B$, i.e., the induction hypothesis is applicable to it.
    By partitioning the incidence homomorphisms from $\Bkl$ to $G$ according to the degree sequence of the vertex that $u$ is mapped to, observe that we have
    \begin{align*}
        \IHom(\Bkl, G) = \sum_{\bar{d} \in \DegreeSeqsOf{n}} \IHomDegreeSeq{B'}{G}{\bar{d}} \cdot \left( \sum_{i = 1}^{n} d_i \cdot i^{s}  \right)^r
    \end{align*}
    for every $r \ge 0$ and every $s > 0$, where for a degree sequence $\bar{d} \in \DegreeSeqsOf{n}$, we let $\IHomDegreeSeq{B'}{G}{\bar{d}}$ denote the number of incidence homomorphisms from $B'$ to $G$ that map $u$ to a vertex with the degree sequence $\bar{d}$.
    
    Let $\bar{d}_1, \dots, \bar{d}_\ell$ be an enumeration of $\DegreeSeqsOf{n}$.
    For every $s > 0$, we obtain the system
    \begin{center}
    \resizebox{1.0\linewidth}{!}{
    \begin{minipage}{\linewidth}
    \begin{align*}
        \begin{pmatrix}
            1 & 1 & 1\\
            \left( \sum_{i = 1}^{n} d_{1,i} \cdot i^{s}  \right)^1 & \hdots & \left( \sum_{i = 1}^{n} d_{\ell,i} \cdot i^{s}  \right)^1\\
            \vdots & \ddots & \vdots\\
            \left( \sum_{i = 1}^{n} d_{1,i} \cdot i^{s}  \right)^{\ell-1} & \hdots & \left( \sum_{i = 1}^{n} d_{\ell,i} \cdot i^{s}  \right)^{\ell-1}
        \end{pmatrix}
        \cdot
        \begin{pmatrix}
            \IHomDegreeSeq{B'}{G}{\bar{d}_1}\\
            \IHomDegreeSeq{B'}{G}{\bar{d}_2}\\
            \vdots\\
            \IHomDegreeSeq{B'}{G}{\bar{d}_{\ell}}
        \end{pmatrix}
        =
        \begin{pmatrix}
            \IHom(B_{0, s}, G)\\
            \IHom(B_{1, s}, G)\\
            \vdots\\
            \IHom(B_{\ell-1, s}, G)
        \end{pmatrix}\\
    \end{align*}
    \end{minipage}
    }
    \end{center}
    of linear equations, where the matrix is the transpose of the Vandermonde matrix $V(\sum_{i = 1}^{n} d_{1,i} \cdot i^{s}, \dots, \sum_{i = 1}^{n} d_{\ell, i} \cdot i^{s})$.
    By choosing $s$ to be large enough, we are able to ensure that these values are pairwise distinct, and thus, that the matrix is invertible:
We choose a large enough $s$ such that we have $\sum_{i = 1}^{j - 1} 2^{n-1} \cdot i^{s} < j^s$ for every $j \in \NumbersTo{n}$, which is certainly possible since we have $\lim_{s \rightarrow \infty} \sum_{i = 1}^{j - 1} i^{s} / j^s = 0$ for every $j \in \NumbersTo{n}$.
    To see that this is in fact sufficient, let $\bar{d}, \bar{d}' \in \DegreeSeqsOf{n}$ with $\bar{d} \neq \bar{d}'$.
    We choose the maximum $j \in \NumbersTo{n}$ such that $d_j \neq d'_j$, where we assume $d_j > d'_j$ without loss of generality.
    Then, we have
    \begin{align*}
        \sum_{i = 1}^{n} d_i \cdot i^s \ge d_j \cdot j^s + \sum_{i = j + 1}^{n} d_i \cdot i^s
        &= d_j \cdot j^s + \sum_{i = j + 1}^{n} d'_i \cdot i^s\\
        &= j^s + (d_j - 1) \cdot j^s + \sum_{i = j + 1}^{n} d'_i \cdot i^s\\
        &> \sum_{i = 1}^{j-1} 2^{n-1} \cdot i^s + (d_j - 1) \cdot j^s + \sum_{i = j + 1}^{n} d'_i \cdot i^s\\
        &\ge \sum_{i = 1}^{j-1} d'_i \cdot i^s + d'_j \cdot j^s + \sum_{i = j + 1}^{n} d'_i \cdot i^s.
    \end{align*}

    By the induction hypothesis, the values $\IHom(B_{0, s}, G), \dots, \IHom(B_{\ell -1, s}, G)$ are determined by $\IHOM{\SBACHs}{G}$, and the invertibility of the matrix yields the same for the value $\IHomDegreeSeq{B'}{G}{\bar{d}}$ for every $\bar{d} \in \DegreeSeqsOf{n}$.
    Because we have
    \begin{align*}
        \IHomDegree{B'}{G}{i} = \sum_{\substack{\bar{d} \in \DegreeSeqsOf{n},\\ \sum_{j = 1}^{n} d_j = i}} \IHomDegreeSeq{B'}{G}{\bar{d}}
    \end{align*}
    for every $i \in \NumbersTo{2^{n-1}}$, this proves our claim and finishes the proof.
\end{proof}

\end{document}